\documentclass[11pt,letterpaper]{article}
\usepackage{amsthm,amsmath,amssymb,amsfonts,graphicx,url}
\usepackage[numbers]{natbib}
\usepackage{amsfonts}
\usepackage{latexsym,nicefrac,bbm}
\usepackage{xspace}
\usepackage[top=1in, bottom=1in, left=1in, right=1in]{geometry}
\usepackage{hyperref} 
\usepackage{multicol}

\usepackage{MnSymbol}

\usepackage{float}
\renewcommand{\epsilon}{\varepsilon}

\newcommand{\nfrac}{\nicefrac}

\newtheorem{theorem}{Theorem}[section]
\newtheorem{lemma}[theorem]{Lemma}

\newtheorem{corollary}[theorem]{Corollary}

\newtheorem*{theorem*}{Theorem}
\newtheorem{remark}[theorem]{Remark}

\usepackage{verbatim}

\newtheorem{problem}{Problem}

\usepackage{algpseudocode}
\usepackage{algorithm,caption}

\usepackage{mhequ}
\def \be{\begin{equs}}
\def \ee{\end{equs}}

\title{\bf Faster polytope rounding, sampling,    and volume computation via a sublinear  ``Ball Walk''}

   \date{}

\usepackage{authblk}

\author[1]{ Oren Mangoubi}
\author[2]{Nisheeth K. Vishnoi}
\affil[1]{\small Worcester Polytechnic Institute}
\affil[2]{\small Yale University}

\begin{document}
\maketitle

\begin{abstract}
This paper studies the problem of ``isotropically rounding" a polytope $K \subseteq \mathbb{R}^n$, that is, computing a linear transformation which makes the uniform distribution on the polytope have roughly identity covariance matrix.  
It is assumed that $K \subseteq \mathbb{R}^n$ is defined by $m$ linear inequalities, with guarantee that $rB \subseteq K \subseteq RB$, where $B$ is the unit ball.  
We introduce a new variant of the ball walk Markov chain and show that, roughly, the expected number of arithmetic operations per-step of this Markov chain is $O(m)$ that is sublinear in the input size $mn$ -- the per-step time of all prior Markov chains.
Subsequently, we apply our new variant of the ball walk to obtain a rounding algorithm that succeeds with probability $1-\epsilon$ in $\tilde{O}(m n^{4.5} \mathrm{polylog}(\frac{1}{\epsilon}, \frac{R}{r}))$ arithmetic operations.
This gives a factor of $\sqrt{n}$ improvement on the previous bound of  $\tilde{O}(m n^{5} \mathrm{polylog}(\frac{1}{\epsilon}, \frac{R}{r}))$, which uses the hit-and-run algorithm. 
  Since the cost of the rounding preprocessing step is in many cases the bottleneck in improving sampling or volume computation running time bounds, our results imply that these tasks can also be achieved in roughly $\tilde{O}(m n^{4.5} \mathrm{polylog}(\frac{1}{\epsilon}, \frac{R}{r}) +mn^4 \delta^{-2})$ arithmetic operations for computing a polytope's volume up to a factor $1+\delta$ and  $\tilde{O}(m n^{4.5} \mathrm{polylog}(\frac{1}{\epsilon}, \frac{R}{r})))$ for generating a sample with TV error $\epsilon$ from the uniform distribution on $K$.  
  This improves on the previous bound of $\tilde{O}(m n^{5} \mathrm{polylog}(\frac{1}{\epsilon}, \frac{R}{r})  +mn^4 \delta^{-2})$ for volume computation in the regime where, roughly, $m\geq n^{2.5}$, and improves on the previous bound of $\tilde{O}(m n^{5} \mathrm{polylog}(\frac{1}{\epsilon}, \frac{R}{r}))$ for generating a sample when, roughly, $m \geq n^{1.5}$.  
  Our algorithm achieves this improvement by a novel method of computing polytope membership, where one avoids checking inequalities which are estimated to have a very low probability of being violated.  We believe that this method is likely to be of independent interest for constrained sampling and optimization problems.
\end{abstract}

\newpage
\tableofcontents

\newpage
\section{Introduction}
The task of bringing a polytope into near-isotropic position is an important problem in mathematics and theoretical computer science (TCS).  
In TCS, this problem is closely linked with the widely studied problem of computing a polytope's volume \cite{dyer1991random, lovasz1993random, kannan1997random, lovasz2006simulated, cousins2017efficient}, and often serves as an important preprocessing step for these algorithms.
 For $a>0,$  we say that a convex body $K$ is in $a$-isotropic position if the uniform distribution on $K$ has covariance matrix $\Sigma_K$ and mean $\mu_K$ satisfying $\frac{1}{a^2} I_n \preccurlyeq \Sigma_{{K}} \preccurlyeq a^2 I_n$ and $\|\mu_{{K}}\|_2 \leq \frac{1}{10}a$. 
Formally,  consider the following problem, where $B$ denotes the unit ball centered at the origin:
\begin{problem}[\bf Bringing a polytope into isotropic position] \label{problem:1}
Given a polytope $K := \{x \in \mathbb{R}^n : A x \leq b\}$, with $A \in \mathbb{R}^{m \times n}$ and $b \in \mathbb{R}^m$ such that $rB \subseteq K \subseteq RB$ for some $r,R>0$, generate a matrix $\tilde{\Omega} \in \mathbb{R}^{n\times n}$ and vector $\tilde{\mu}\in \mathbb{R}^n$ such that $\tilde{K} :=\tilde{\Omega}^{-\frac{1}{2}} (K -\tilde{\mu})$ is in 2-isotropic position.
\end{problem}
\noindent
The problems of sampling from the uniform distribution on a polytope and of bringing a polytope into isotropic position are closely related. 
 On the one hand, bringing a polytope into isotropic position can improve the running time of Markov chain-based sampling algorithms \cite{lee2018kannan}.  
 On the other hand, it is known that $n\log(n)$ independent samples from the uniform distribution on a polytope suffices to bring a polytope into $O(1)$-isotropic position \cite{rudelson1999random}.  
 All current Markov chains which are used to sample from the uniform distribution on a polytope defined by $m$ inequalities use at least $mn$ arithmetic operations per Markov chain step to implement, and it is currently an open problem how to improve the number of arithmetic operations to fewer than $mn$ \cite{lee2018kannan}.  
$mn$ is the size of the input and is also the time required to check whether a given point is in $K$ or not.

The main focus of this paper is to develop  Markov chains for sampling that allow us to bypass this $mn$ barrier, and obtain faster algorithms for rounding, sampling, and volume computation.
In particular, we introduce an implementation of the ball walk Markov chain \cite{applegate1991sampling, kannan1997random}, which improves the expected number of arithmetic operations to roughly $O(m)$ operations per ball walk step.  
Our improvement in the per-step complexity applies in the special case when the polytope is in near-isotropic position, and we are given an $O(1)$-warm start in the $n^{-3}$ interior of the polytope.

Key to our results is a new variant of the ball walk that, given an O(1)-warm start $X_0 \in K$, requires only  roughly $O(m)$  expected number of operations per step (after the first step).  
 We then apply a recent result of \cite{lee2018kannan} which says that, starting at any point on the $n^{-3}$ interior of $K$, the ball walk, together with a rejection sampling post-processing step from \cite{kannan1997random}, can generate a sample from the uniform distribution on $K$ in $O(n^{2.5} \log(\frac{1}{\epsilon}))$ ``proper'' ball walk steps (that is, only counting the steps where the ball walk changes position).  
 If $X_0$ is also O(1)-warm, the expected number of steps (both proper and improper) is also $O(n^{2.5} \log(\frac{1}{\epsilon}))$ \cite{kannan1997random}.  
 Multiplying the two, we get that the expected number of operations for our implementation of the ball walk to generate a sample with TV error $\epsilon$, is, roughly speaking\footnote{Since the two random variables may be correlated, we cannot simply multiply their expectations.  Instead, we treat these two expectations separately until the very end of our proof, and then use Markov's inequality to bound each with probability $9/10$.  We then multiply our bounds for the two random variables that we obtained from Markov's inequality to get a bound for the product of the random variables which holds with probability 8/10.}, $O(m \times n^{2.5} \log(\frac{1}{\epsilon}))$.
Therefore, if we re-start the ball walk at the same $O(1)$-warm initial point $X_0$ after we generate each sample, we can use our algorithm to generate $p$ samples that are (conditionally on $X_0$) jointly independent and uniformly distributed on $K$ with TV error $\epsilon$, after roughly $O(p mn^{2.5} \log(\frac{1}{\epsilon}))$ operations. 

Using our implementation of the ball walk to generate $n\log(n)$ samples, we can use results from \cite{rudelson1999random} to compute a sample mean and sample covariance matrix for the uniform distribution on $K$ which allows us to bring any polytope $K$ that is in $15$-isotropic position into $2$-isotropic position with probability $1-\epsilon$, in roughly $\tilde{O}(mn^{3.5} \log(\frac{1}{\epsilon}))$ operations, if we are given an $O(1)$-warm start $X_0$, with $X_0$ in the $n^{-3}$ interior of $K$.
 We use this idea in an iterative manner to obtain a rounding algorithm which can bring any polytope $rB \subseteq K \subseteq RB$ into $2$-isotropic position with probability $1-\epsilon$.
  Specifically, we show the following:
 \begin{theorem}[\bf Main theorem: Bringing a polytope into isotropic position] \label{thm:main_intro}
There exists an algorithm which, given a polytope $K := \{x \in \mathbb{R}^n : A x \leq b\}$, with $A \in \mathbb{R}^{m \times n}$ and $b \in \mathbb{R}^m$ such that $rB \subseteq K \subseteq RB$ for some $r,R>0$, and $\epsilon>0$, generates a matrix $\tilde{\Omega} \in \mathbb{R}^{n\times n}$ and vector $\tilde{\mu}\in \mathbb{R}^n$ such that the polytope $\tilde{K} :=\tilde{\Omega}^{-\frac{1}{2}} (K -\tilde{\mu})$ is in 2-isotropic position with probability at least $1-\epsilon$, in $\tilde{O}(mn^{4.5} \mathrm{polylog}(\frac{1}{\epsilon}, \frac{R}{r}))$ arithmetic operations.
\end{theorem}
\noindent
  To prove Theorem \ref{thm:main_intro} (see Theorem \ref{thm:main_AlgorithmSpecific} for the version of Theorem \ref{thm:main_intro} specific to our algorithm), we consider a sequence of convex bodies $K_i := K \cap (1+\nicefrac{1}{n})^i rB$, and bring these convex bodies into isotropic position sequentially, starting from $K_0 = rB$.  
 We are able to do this since one can show that the same linear transformation which brings $K_i$ into 2-isotropic position also brings $K_{i+1}$ into $15$-isotropic position.  
 Since there are $n\log(\frac{R}{r})$ convex bodies $K_i$ in the sequence, our algorithm brings $K$ into $2$-isotropic position in $\tilde{O}(mn^{4.5} \mathrm{polylog}(\frac{1}{\epsilon}, \frac{R}{r}))$ operations.  
 Our rounding algorithm improves the best previous $\tilde{O}(mn^5 \mathrm{polylog}(\frac{1}{\epsilon}, \frac{R}{r}))$ bound of \cite{lovasz2006simulated} for bringing a polytope into isotropic (or just ``well-rounded'') position by a factor of $\sqrt{n}$, making progress on an open problem
 (see \#4 in Section 8 of \cite{cousins2017efficient}).  
We get an improvement of $\sqrt{n}$ rather than $n$ since our bound for the number of operations per ball walk step needs the convex bodies $K_i$ to be kept in isotropic position at each $i$ (requiring us to generate $\tilde{\Theta}(n^2)$ independent samples), while  \cite{lovasz2006simulated} only need to keep their sequence of convex bodies in well-rounded position at each iteration (which they can do using only $\tilde{\Theta}(n)$ independent samples).
  On the other hand, each of our samples requires $n^{1.5}$ fewer operations per sample in expectation: we get a factor of $n$ fewer operations from our improved bound on the expected number of operations per Markov chain step, and an additional factor of $\sqrt{n}$ fewer operations because the bound on the number of ball walk steps on isotropic convex bodies is smaller by a factor of $\sqrt{n}$ than the bound for the hit-and-run Markov chain used in \cite{lovasz2006simulated}. 
  Our bound on the number of operations to put $K$ in isotropic position is therefore smaller by a factor  $\sqrt{n}$ compared to  \cite{lovasz2006simulated}.

\paragraph{Application to volume computation.}
Bringing K into isotropic position allows us to then use the volume computation algorithm of \cite{cousins2017efficient} to compute the volume of $K$ with error $\delta$, in $\tilde{O}(\frac{mn^4}{\delta^2})$ operations after we pre-process $K$ into isotropic position.  
Hence, starting with $K$ far from isotropic position, we can compute the volume of $K$ in roughly $\tilde{O}(mn^{4.5} \log(\frac{R}{r}) + \frac{mn^4}{\delta^2})$ operations:
 \begin{corollary}[\bf Computing the volume of a polytope] \label{thm:intro_volume}
There exists an algorithm which, given a polytope $K := \{x \in \mathbb{R}^n : A x \leq b\}$, with $A \in \mathbb{R}^{m \times n}$ and $b \in \mathbb{R}^m$ such that $rB \subseteq K \subseteq RB$ for some $r,R>0$, and $\epsilon, \delta>0$, computes with probability at least $1-\epsilon$ the volume of $K$ up to a factor of $1+\delta$ in $\tilde{O}(mn^{4.5} \mathrm{polylog}(\frac{1}{\epsilon}, \frac{R}{r}) + \frac{mn^4}{\delta^2}\mathrm{polylog}(\frac{1}{\delta}, \frac{1}{\epsilon}) )$ arithmetic operations.
\end{corollary}
\noindent In the regime where $\delta^{-1} = O(1)$ and $m \geq n^3$, the best current algorithm, which uses \cite{lovasz2006simulated} for preprocessing and \cite{cousins2017efficient} for volume computation, gives a bound of $\tilde{O}(mn^{5} \mathrm{polylog}(\frac{1}{\epsilon}, \frac{R}{r}) + \frac{mn^4}{\delta^2}\mathrm{polylog}(\frac{1}{\delta}, \frac{1}{\epsilon}) )$ operations.  
Corollary \ref{thm:intro_volume} improves this bound by a factor of $\sqrt{n}$.  
Moreover, since our result benefits from recent improvements in the bound on the Cheeger constant of an isotropic convex body, Corollary \ref{thm:intro_volume} makes progress towards the open problem of connecting improved bounds on the Cheeger constant to faster volume computation (See section 2.2.3 of \cite{lee2018kannan}).
In table \ref{table:Volume} we give bounds for different algorithms which can be used to compute the volume.   
\paragraph{Application to sampling.}
Preprocessing a polytope into isotropic position is also a bottleneck for the problem of sampling from the uniform distribution on the polytope.  
If we use our rounding algorithm (Theorem \ref{thm:main_intro}) to bring $K$ into $2$-isotropic position, and then use the hit-and-run algorithm to generate a sample from the uniform distribution on $K$, we obtain a sample from the uniform distribution on $K$ with TV error $\epsilon$  in $\tilde{O}(mn^{4.5} \mathrm{polylog}(\frac{1}{\epsilon}, \frac{R}{r}))$ arithmetic operations:
 \begin{corollary}[\bf Sampling from  non-rounded polytope] \label{cor:sampling}
There exists an algorithm which, given a polytope $K := \{x \in \mathbb{R}^n : A x \leq b\}$, with $A \in \mathbb{R}^{m \times n}$ and $b \in \mathbb{R}^m$ such that $rB \subseteq K \subseteq RB$ for some $r,R>0$, and $\epsilon>0$, generates a sample uniformly distributed on $K$ with TV error $\epsilon$, in $\tilde{O}(mn^{4.5} \mathrm{polylog}(\frac{1}{\epsilon}, \frac{R}{r}))$ arithmetic operations.
\end{corollary}
\noindent In the regime where $m \geq n^2$, the current best bound is $\tilde{O}(mn^{5}\mathrm{polylog}(\frac{1}{\epsilon}, \frac{R}{r}))$ operations (if one uses a rounding pre-processing step from \cite{lovasz2006simulated}, and then the hit-and-run Markov chain \cite{lovasz2006hit}) or, depending on the matrix multiplication exponent $2< \omega \leq 3$, $\tilde{O}(mn^{\omega +2.5}\log(\frac{1}{\epsilon}))$ for the John walk \cite{chen2018fast}.  
Corollary \ref{cor:sampling} improves on \cite{lovasz2006simulated} by a factor of $\sqrt{n}$, and on \cite{chen2018fast} by a factor of $n^{\omega -2}$.  However, for smaller values of $m$ algorithms such as Riemannian HMC \cite{lee2017convergence} can be faster (see table \ref{table:RandomWalks}).

\begin{table}[H]
\begin{center}
\begin{small}
\begin{tabular}{|c|c|c|}
\hline 
Algorithm
&  number of arithmetic operations
\tabularnewline
\hline
 Ball walk + rounding \cite{kannan1997random}
  &    $m n^{6} \delta^{-2}$%

\tabularnewline
\hline
 Ball walk + rounding \cite{lee2017eldan}  + Gaussian cooling
  &    $m n^{5.5}   + mn^4 \delta^{-2}$
\tabularnewline
\hline 
Hit-and-run + simulated annealing  \cite{lovasz2006simulated}   
 &    $m n^{5} \delta^{-2}$
      \tabularnewline
\hline 
Riemannian HMC \cite{lee2017convergence}
 &    $m^2 n^{\omega-\frac{1}{3}}  \delta^{-2}$
      \tabularnewline
\hline 
Gaussian cooling \cite{cousins2015bypassing}
 &    $m n\times \max(n^2 (\frac{R}{r})^2, n^3) \delta^{-2}$
  \tabularnewline
\hline 
\textbf{Algorithm 1 \& 2 [our paper]} 
 + Gaussian cooling
 &    $m n^{4.5} + m n^4 \delta^{-2}$
\tabularnewline
\hline 
\end{tabular}
\end{small}
\end{center}
\vspace{-3mm}
\caption{{\small Bounds on the number of arithmetic operations to compute the volume of a polytope $K$ with $rB \subseteq K \subseteq RB$ (logarithmic factors of $r, R, \epsilon, d, m$ are not shown). Here $\omega$ is the matrix multiplication exponent, currently $\omega \approx 2.37$.  (note: Gaussian cooling assumes $(r,R)$-well rounded, which is somewhat weaker than $rB \subseteq K \subseteq RB$)}}
\label{table:Volume}
\end{table}

\vspace{-2mm}
\begin{table}[H]
\begin{center}
\begin{small}
\begin{tabular}{|c|c|c|}
\hline 
Algorithm
&  number of arithmetic operations
\tabularnewline
\hline
 Ball walk + rounding \cite{lee2017eldan}
  &   $mn^{5.5}$
\tabularnewline
\hline 
Hit-and-run + rounding  \cite{lovasz2006simulated}
 &    $mn^{5}$
            \tabularnewline
\hline 
Dikin walk \cite{kannan2012random}
    
 &    $m^2 n^{\omega+1}$
\tabularnewline
\hline 
John walk + Dikin walk initialization \cite{chen2018fast}
 &    $m n^{\omega+2.5}$
   \tabularnewline
\hline 
Geodesic walk* \cite{lee2017geodesic}
 &   $m^2 n^{\omega+\frac{3}{4}} $\\ & ($m^2 n^{\omega-\frac{1}{4}} \log(\beta)$ assuming a $\beta$-warm start)
   \tabularnewline
\hline 
Riemannian HMC* \cite{lee2017convergence}
 &    $m^2 n^{\omega+\frac{2}{3}} $\\
 & ($m^2 n^{\omega-\frac{1}{3}} \log(\beta)$ assuming a $\beta$-warm start)
 \tabularnewline
\hline 
Vaidya walk + Dikin walk initialization \cite{chen2018fast}
 &    $m^{1.5} n^{\omega+\frac{3}{2}}$
  \tabularnewline
\hline 
\textbf{Algorithm 1 \& 2 [our paper]} 
 &    $m n^{4.5}$ 
\tabularnewline
\hline 
\end{tabular}
\end{small}
\end{center}
\vspace{-3mm}
\caption{\small Bounds on the number of arithmetic operations to generate one sample from the uniform distribution on a polytope $K$ with $rB \subseteq K \subseteq RB$ (logarithmic factors of $r, R, \epsilon, d, m$ are not shown).
The matrix multiplication exponent $\omega$ is currently $\omega \approx 2.37$ for the best known matrix multiplication algorithm. (*Note: RHMC and Geodesic walk assume a warm start.  While one can obtain a $(\nicefrac{R}{r})^n$-warm start in our setting by initializing from the uniform distribution on $rB$, this causes the running times of RHMC and Geodesic walk to gain an additional factor of $n$.)}
\label{table:RandomWalks}
\end{table}
\noindent
Note that in table \ref{table:RandomWalks} the Dikin, Geodesic and Vaidya walks, as well as Riemannian HMC, have bounds with dependence on $m$ of at least $m^{1.5}$. 
 In particular, in the regime $m>n^2$, they have slower bounds than the algorithms which have linear dependence on $m$, including our algorithm as well as the hit-and-run algorithm of \cite{lovasz2006simulated} and the ball walk of \cite{lee2017eldan}.  
Finally, note also that, given the current bound $\omega \approx 2.37,$ our algorithm has faster bounds when compared to the John walk.

\paragraph{Key technical ideas.} The algorithmic techniques we use in our variant of the ball walk are inspired from stochastic gradients, where one queries an oracle by subsampling.  
Instead of subsampling a small subset of component gradients, at each step of the ball walk we check a small subset of the inequalities defining our polytope.  
The challenge is in determining which inequalities are important at any given time.
The main difficulty lies in the fact that the ball walk is much more likely to violate inequalities corresponding to nearby hyperplanes than far-away ones.
The reason is that, if the Markov chain's steps are uniformly distributed, by the isoperimetric inequality \cite{osserman1978isoperimetric} (and convexity of the polytope) the Markov chain will in expectation spend at least roughly half of its time a distance of $\frac{1}{n}$ from the boundary of the polytope, if the polytope is in near-isotropic position.  
 Hence, in expectation, the Markov chain will be a distance of $\frac{1}{n}$ from one or more faces of the polytope at least half of the time and one cannot simply check a uniform random subset of these inequalities.  
Instead, our Markov chain estimates the distance to each inequality, and checks only those inequalities which have a non-negligible probability of being violated. 
This  idea turns out to be challenging to execute and we expand on it in Section \ref{sec:technical-overview}. 
 We believe our method of checking membership may be of independent interest to constrained sampling and optimization problems.

\section{Our algorithms}
\paragraph{Notation.}
For any finite-volume subset $S \subseteq \mathbb{R}^n$, let  $\Sigma_S$ denote the covariance matrix of the uniform distribution on $S$ and let $\mu_S$ denote the mean of the uniform distribution on $S$.
Let $\pi_K$ denote the uniform distribution on   a convex body $K$. 
Let $\hat{\pi}_K \equiv \hat{\pi}_K^\eta$ denote the speedy distribution for step size $\eta$, where the speedy distribution is the stationary distribution of the proper steps of the ball walk with step size $\eta$.  
The Markov chain formed by the proper steps of the ball walk is called the speedy walk.
A convex body $K$ is said to be $a$-isotropic if $\frac{1}{a^2} I_n \preccurlyeq \Sigma_K \preccurlyeq a^2 I_n$.  
If, furthermore, $\|\mu_K\|_2 \leq \frac{1}{10} a$,  we say that $K$ is in $a$-isotropic position.
A convex body $K$ is said to be $(r,\mathfrak{R})$-rounded if $\mathbb{E}_{X\sim \pi_K}[\|X\|_2^2] \leq  \mathfrak{R}^2$ and $K$ contains a ball of radius $r$. 
 A convex body which is $O(1)$-isotropic is also $(O(1),O(\sqrt{n}))$-rounded, although the converse is not true \cite{lovasz2006fast}.  If a convex body is $(r, \mathfrak{R})$-rounded with $\frac{\mathfrak{R}}{r} = O(\sqrt{n})$, we say it is \textit{well-rounded}.
For any subset $S\subseteq \mathbb{R}^n$ and any point $x\in \mathbb{R}^n$, let $\mathrm{dist}(x, S) := \mathrm{\inf_{y\in S}} \|x-y\|_2$ denote the distance from $x$ to the subset $S$ and let $\partial S$ denote the boundary of the subset $S$.
Let $B(x,r)$ denote the Euclidean ball with center $x$ and radius $r$, and denote the unit ball by the shorthand notation $B:= B(0,1)$.
We say that a probability distribution $\mu:\mathbb{R}^n \rightarrow \mathbb{R}$ (or a random variable with distribution $\mu$) is $\beta$-warm with respect to a probability distribution $\nu:\mathbb{R}^n \rightarrow \mathbb{R}$ if $\frac{\mu(x)}{\nu(x)}\leq \beta$ for all $x \in \mathbb{R}^n$. 
 We denote the probability distribution of a random variable $X$ by $\mathcal{L}(X)$.  Finally, we define the ball walk Markov chain $\tilde{X}_0,\tilde{X}_1\ldots$ on $K$ with initial point $\tilde{X}_0 \in K$ and step size $\eta>0$, by the recursion $\tilde{X}_{i+1} = \tilde{X}_{i} + \eta \xi_i$ if $\tilde{X}_{i} + \eta \xi_i \in K$ and $\tilde{X}_{i+1} = \tilde{X}_{i}$ otherwise, where $\xi_0, \xi_1,\ldots$ are iid uniform on the unit ball.

\paragraph{Algorithm 1 (Sampling).} Algorithm \ref{alg:BallModified} generates independent samples approximately uniform on a convex body $K := \{x \in \mathbb{R}^n : A x \leq b\} \cap \rho B$ for some $\rho>0$. 
 It has two main components: a Markov chain obtained as a subsequence of the ball walk which generates samples from the ``speedy distribution", and a rejection sampling method which obtains uniform samples from these ``speedy distributed" samples. 
  The ``While" loop generates the ball walk $X_1,\ldots $. 
   To determine whether a ball walk proposal is inside $K$ without checking all $m$ inequalities at each step, every time that the 
 \begin{algorithm}[H]
\caption{Sampling (Modified speedy walk) \label{alg:BallModified}}
\flushleft
\vspace{-4mm}
\textbf{input:} $A \in \mathbb{R}^{m \times n}$, $b \in \mathbb{R}^m$, $\rho>0$, with $\frac{1}{10} B \subseteq K$ where $K := \{x \in \mathbb{R}^n : A x \leq b\} \cap \rho B$\\
\textbf{input:}  step size $\eta>0$, tolerance $\alpha>0$, $\mathcal{I}>0$, maximum number $i_{\mathrm{max}}$ of proper+improper steps \textbf{input:}   Initial point $Y_0 \in K := \{x \in \mathbb{R}^n : A x \leq b\} \cap \rho B$, $\textrm{Modified} \in \{\textrm{ON}, \textrm{OFF}\}$\\
 \textbf{output:} $p$ approximately independent samples from the uniform distribution on $K$.
\begin{algorithmic}[1]
\For{$j = 1$ to $m$} \Comment{Initialization: compute distance to all $m$ hyperplanes}
\State Set $h_j = b_j - A_j X_0$
\EndFor
\State Sort the $h_j$'s in increasing order, and denote the ordered set of $h_j$'s by $H$
\State Set $i=0$, $s=0$, and $k=1$
\While{$k\leq p$}
\State Set $X_0= Y_0$
\While{$s \leq \mathcal{I}$ and $i\leq i_{\mathrm{max}}$}
\State Set $i \leftarrow i+1$, and set $\mathsf{K} = \textrm{True}$
\State Sample $\xi_i \sim \mathrm{unif}(B(0,1))$
\State Set $\hat{X}_{i+1} = X_i + \eta \xi_i$ \Comment{Ball walk proposal}
\If{\textrm{Modified} = \textrm{ON}} \Comment{Determining which inequalities to check}
\State Set $j^\star$ to be the largest value of $j$ such that $h_j < \alpha \frac{\eta}{\sqrt{n}} \times i$
\For{$j=1$ to $j^\star$}  \Comment{Compute distance to possibly-nearby hyperplanes}
\State Set $h_j = b_j - A_j X_i$ 
\State Insert $h_j$ into $H$ such that $H$ remains in increasing order
\If{$h_j < 0$ and  $\mathsf{K} = \textrm{True}$}
\State Set $\mathsf{K} = \textrm{False}$
\EndIf
\EndFor
\If{$\|\hat{X}_{i+1}\|_2 > \rho$}   \Comment{Check if in ball $\rho B$}
\State Set $\mathsf{K} = \textrm{False}$
\EndIf
\Else  \Comment{Conventional membership query, if modifications ``turned off"}
\State Check if $\hat{X}_{i+1}$ satisfies all $m$ inequalities and is in $\rho B$, and if not set $\mathsf{K} \leftarrow \textrm{False}$
\EndIf
\If{$\mathsf{K}= \textrm{True}$}
\State Set $Y_{s} = X_{i}$   \Comment{Speedy walk step}
\State Set $X_{i+1} = \hat{X}_{i+1}$
\State Set $s \leftarrow s+1$
\Else
\State Set $X_{i+1} = X_i$
\EndIf
\EndWhile
\State Set $Z_k = \frac{2n}{2n-1} Y_{s}$
\If{$Z_k \in K$}
\State ``accept" $Z_k$ and set $k \leftarrow k+1$
\Else
\State ``reject" $Z_k$  \Comment{rejection sampling to get uniform distribution}
\EndIf
\EndWhile
\State \textbf{Output:} Samples $Z_1,\ldots, Z_p$ which are approximately uniformly distributed.  Output these samples if $i \leq i_{\mathrm{max}}$.  Otherwise, output ``Failure".
\end{algorithmic}
\end{algorithm}
\noindent algorithm checks an inequality $A_j x \leq b$, it stores in memory the distance $h_j$ to the corresponding hyperplane $H_j$. 
   The algorithm then waits $\frac{h_j}{\alpha \nfrac{\eta}{\sqrt{n}}}$ steps until re-computing that inequality, where $\alpha>0$ is a parameter set by the user. 
    The idea is that, since with high probability the ball walk makes steps of size $O(\nfrac{\eta}{\sqrt{n}})$  in  the direction of the hyperplane $H_j$, the ball walk is unlikely to  propose a step which crosses $H_j$ before taking $\frac{h_j}{\alpha \nfrac{\eta}{\sqrt{n}}}$ steps. 
     This allows us to ensure that our implementation of the ball walk remains inside $K$ with high probability.

The ball walk is run until a fixed number of proper steps $Y_s$ are made. 
 The Markov chain $Y_1, Y_2,\ldots$ formed by the proper steps of the ball walk is called the ``speedy walk".  
 Since fast mixing bounds are available for the speedy walk but not for the ball walk, we generate our samples from the speedy walk, that is, we run the ball walk for a fixed number $\mathcal{I}$ of proper steps.  
 This gives us a sample $Y_\mathcal{I}$ approximately from the stationary distribution of the speedy walk. 
 Unfortunately, the 
 \begin{algorithm}[H]
\caption{Rounding} \label{alg:rounding}
\flushleft
\textbf{input:} $A \in \mathbb{R}^{m \times n}$, $b \in \mathbb{R}^m$\\
\textbf{input:} $r,R>0$ such that $r B\subseteq K \subseteq R B$, where $K := \{x \in \mathbb{R}^n : A x \leq b\}$\\
\textbf{input:} $p\in \mathbb{N}$, $\epsilon>0$, $\textrm{Modified} \in \{\textrm{ON}, \textrm{OFF}\}$
\begin{algorithmic}[1]
\State Set $\hat{\Sigma}_0 = r^2 I_n$ and $\hat{\mu}_0=0$
\State Set $i^\star= n\log_2(\frac{R}{r})$
\For{$k=1$ to $p$}
\State Sample $Z_k$ from the uniform distribution on $B$
\If{$\mathrm{dist}(Z_k, \partial B) \geq n^{-3}$}
\State Set $\hat{X}_0 = Z_k$
\EndIf
\EndFor
\For{$i = 1$ to $i^\star-1$}
\State define $K_i := (1+\nicefrac{1}{n})^i rB \cap K$ (just a definition, no computation here)
\State Use Algorithm \ref{alg:BallModified} with parameter ``$\mathrm{Modified}$" and initial point $X_0$ to generate $p$ points $Z_1,\ldots, Z_p$ approximately from the uniform distribution on $\hat{K}_i \cap 20 \sqrt{n}\log(\frac{40 n^2}{\epsilon})B$, where $\hat{K}_i := \hat{\Sigma}_i^{-\frac{1}{2}} (K_i - \hat{\mu}_i)$ and is represented by the inequalities with matrix $A \hat{\Sigma}_i^{\frac{1}{2}} \in \mathbb{R}^{m \times n}$, and vector $b -A \hat{\mu}_i \in \mathbb{R}^m$.
\State Set $\hat{\mu}_{i+1} = \frac{1}{p}\sum_{j=1}^p Z_i + \hat{\mu}_i$
\State Set $\hat{\Sigma}_{i+1} = [\frac{1}{p} \sum_{j=1}^p (Z_i-\mu_i)^\top (Z_i-\mu_i)] \hat{\Sigma}_{i}$
\State Set $\textrm{Interior} = \textrm{False}$
\While{ $\textrm{Interior} = \textrm{False}$} \Comment{generate a starting point, uniform on the ``$n^{-3}$-interior" of $\hat{K}_i$}
\State Use Algorithm \ref{alg:BallModified}  with parameter ``$\mathrm{Modified}$" and initial point $X_0$ to a generate a single sample $\hat{X}_0$ approximately from the uniform distribution on $\hat{K}_i \cap 20 \sqrt{n}\log(\frac{40 n^2}{\epsilon})B$.
\State Set $\hat{X}_0' \leftarrow \hat{\Sigma}_{i+1}^{-\frac{1}{2}} \hat{\Sigma}_{i}^{\frac{1}{2}}(\hat{X}_0 + \hat{\mu}_i - \hat{\mu}_{i+1})$
\If{$\mathrm{dist}(\hat{X}_0', \partial [\tilde{K}_i \cap 20 \sqrt{n}\log(\frac{40 n^2}{\epsilon})B]) \geq n^{-3}$, where $\tilde{K}_{i} := \hat{\Sigma}_{i+1}^{-\frac{1}{2}}(K_{i}-  \hat{\mu}_{i+1})$}
\State Set $X_0 = \hat{X}_0'$
\State Set  $\textrm{Interior} = \textrm{True}$
\EndIf
\EndWhile
\EndFor
\State \textbf{output:}  $\hat{\Sigma}_{i^\star}, \hat{\mu}_{i^\star}, X_0$
\end{algorithmic}
\end{algorithm}

\noindent  speedy walk does not have uniform stationary distribution; the samples from the speedy walk have a different distribution called the ``speedy distribution". 
 To obtain uniformly distributed samples from our speedy-distributed samples we use a rejection sampling method, Algorithm 4.15 from \cite{kannan1997random} (reproduced in our Algorithm \ref{alg:BallModified} as Steps 35-40).  
 To obtain $p$ independent samples, we run the speedy walk $p$ times starting at the same initial point $Y_0$ but using different independent random Gaussian vectors $\xi_i$ each time (this is the outer ``For" loop). 
 The parameter ``$\textrm{Modified}$" can be set to ``$\textrm{ON}$" or ``$\textrm{OFF}$".  If it is ``$\textrm{ON}$" we use our new implementation of the ball walk, while if it is ``$\textrm{OFF}$" we use the usual implementation where all $m$ inequalities are checked at each step. 
  The purpose of the parameter ``$\textrm{Modified}$" is only to simplify the exposition of our proofs; in practice we always set``$\textrm{Modified}=\textrm{ON}$."

\paragraph{Algorithm 2 (Rounding).} 
Using Algorithm \ref{alg:BallModified} as a subroutine, we obtain an algorithm (Algorithm \ref{alg:rounding}) for bringing a polytope into isotropic position:

The goal of Algorithm \ref{alg:rounding} is to inductively bring a sequence of convex bodies $K_1\subseteq K_2\subseteq \cdots$, where $K_i := (1+\nicefrac{1}{n})^i rB \cap K$, into isotropic position, starting with $K_1$. 
At each iteration of the ``For" loop, we use Algorithm \ref{alg:BallModified} to generate samples from the uniform distribution on a convex body $\hat{K}_i$ which is an affine transformation of $K_i$ in $15$-isotropic position obtained at the previous step of the ``For" loop (Step 11).  
Using these samples, Algorithm \ref{alg:BallModified} computes a sample covariance matrix and mean for $\hat{K}_i$ (Steps 12 and 13), which allows it to compute an affine transformation that puts $K_i$ into $2$-isotropic position as ``$\tilde{K}_i$"  and $K_{i+1}$ into $15$-isotropic position as ``$\hat{K}_{i+1}$". 
 Steps 15-21 generate a point $X_0$ which is O(1)-warm with respect to the uniform distribution and in the $n^{-3}$-interior of $\tilde{K}_i$. 
  The point $X_0$,  which is also $O(1)$-warm with respect to the uniform distribution on  $\hat{K}_{i+1}$ and in its  $n^{-3}$-interior, is then used in the next iteration of the ``For" loop as an initial point when Algorithm \ref{alg:BallModified} is used to generate samples from the uniform distribution on $\hat{K}_{i+1}$.

\section{Technical overview of our main result: Theorem \ref{thm:main_intro}}\label{sec:technical-overview}
\paragraph{Rounding polytopes via sampling.}
\medskip
Most algorithms which bring a polytope into isotropic position work by generating independent or near-independent samples which are approximately uniformly distributed in the polytope.  
These samples allow one to compute the sample mean and sample covariance matrix for the polytope.  In \cite{rudelson1999random} it was shown that $n\log(n)$ samples suffice to bring a polytope into isotropic position.  
However, one is still left with the problem of generating independent uniform samples from the polytope.  
Typically this is done by running a Markov chain on the polytope whose stationary distribution is equal (or in some sense close to) the uniform distribution in the polytope.  
However, the number of steps for which one must run the Markov chain to obtain uniform independent samples in many cases itself depends on the extent to which the polytope is isotropic, or the extent to which it is rounded.\footnote{While there are some algorithms such as the Dikin walk which do not depend on how isotropic the polytope is, these do not currently provide the fastest methods of bringing a polytope $rB \subseteq K \subseteq RB$ into isotropic position.} 

\smallskip
Because of this, most rounding algorithms start with the ball contained in the polytope (which is very easy to put in well-rounded position), and gradually deform the polytope at each iteration (for instance by considering the intersection of the polytope with a ball of increasing radius).  
One alternates between sampling steps where one samples from the convex body, and steps where one uses these samples to compute an affine transformation which keeps the convex body well-rounded (for example, this transformation can be achieved by computing the sample covariance matrix).  
For instance, this is the case for the ``ball walk'' Markov chain used in \cite{kannan1997random}, and the ``hit-and-run'' Markov chain used in \cite{lovasz2006simulated}.  
In particular, the algorithm of \cite{lovasz2006simulated} requires only $n \log(n)$ samples to round the convex body.  
Key to this is that the hit-and-run Markov chain does not require isotropic position but rather only that the polytope be well-rounded. 
This fact was used in \cite{lovasz2006simulated}, together with a ``pencil construction," to provide a rounding algorithm where one computes a well-rounded polytope at $n\log(\frac{R}{r})$ iterations each using $\log(n)$ samples, and using $n$ samples to bring the polytope into isotropic position only every $\log(n)$ iterations.  
Since the hit-and-run Markov chain requires $n^3$ steps to generate a uniform sample, they require roughly $n^4 \log(n)\log(\frac{R}{r})$ Markov chain steps to put a polytope into isotropic position.  
If the polytope is defined by $m$ inequalities, the hit-and-run algorithm uses $mn$ arithmetic operations to compute polytope membership at each step of the Markov chain, giving a bound of $mn \times n^4 \log(n))\log(\frac{R}{r})$ arithmetic operations to round the polytope; this is currently the fastest running time bound for rounding this class of polytopes.
 
  \smallskip
 This $mn$ cost of computing each step of the Markov chain is a feature of all current Markov chain sampling algorithms on polytopes defined by $m$ inequalities \cite{lee2018kannan}. 
  However, one can imagine that there may be ways of reducing the cost of computing polytope membership. 
   One approach is to use a Markov chain called ``coordinate hit-and-run" \cite{haraldsdottir2017chrr}.  
   This algorithm works in the same way as the usual hit-and-run algorithm, except that it only takes steps in (random) coordinate directions.  
   Hence,  checking each inequality takes only $O(1)$ arithmetic operations, meaning that each step of the Markov chain would roughly require only $m$ operations.  
   Unfortunately, since there are as of yet no polynomial-in-dimension mixing time bounds for coordinate hit-and-run, one cannot currently use coordinate hit-and-run to obtain better running time bounds for rounding.

 \paragraph{A first attempt.}
 As an alternative approach to coordinate hit-and-run one might consider using the stochastic gradient technique to reduce the cost of computing each step of a Markov chain which stays inside a polytope.  
 For instance, one might attempt to apply stochastic gradients to a Markov chain such as the Dikin walk which, instead of computing polytope membership, makes use of the log-barrier function of the polytope to remain inside the polytope.  
 The log-barrier at any point in the Markov chain is given by $\phi(x) = -\sum_{j=1}^m \log(A_j x - b_j)$ with Hessian $\nabla^2 \phi(x) = -\sum_{j=1}^m \frac{A_j^\top A_j}{(A_j x - b_j)^2}$ (here $A_j$ is a row-vector). 
  If $m$ is large, one might try to estimate the gradient, or in the case of the Dikin walk the Hessian, of the log-barrier function by taking a small subset of the polytope's inequalities and using these to estimate the sum.  
  Unfortunately this ``stochastic Hessian" gives a very bad approximation for points which are near a face of the polytope, since the term $\frac{1}{(A_j x - b_j)^2}$ corresponding to the nearest face can be much larger than the combined contributions of the terms corresponding to all the other faces, and a small subsample of the polytope inequalities will most likely not include the single very large term.  
  Since a uniform random sampling is unlikely to include this overwhelmingly large term in the sum, this ``stochastic Hessian" version of the Dikin walk is likely to very quickly leave the polytope.
 
  \smallskip
One might instead consider an approach related to stochastic gradients, but for Markov chains such as the ball walk which, instead of computing a barrier function, compute polytope membership at each step.  
To determine polytope membership, one typically checks all $m$ inequalities at each step of the Markov chain. 
 One may instead consider checking only a small uniformly random subset of $\frac{m}{n}$ of these inequalities at each step.  
 Unfortunately, this approach cannot work  if one wishes to sample from the uniform distribution on a polytope.  
 The reason is that, if the Markov chain's steps are uniformly distributed, by the isoperimetric inequality  \cite{osserman1978isoperimetric} (and convexity of the polytope) the Markov chain will in expectation spend at least half of its time a distance of $\frac{1}{n}$ from the boundary of the polytope.  
 Hence, in expectation, the Markov chain will be a distance of $\frac{1}{n}$ from one face of the polytope at least half of the time. 
 For a ball walk with optimal step size of $\eta = \Theta(\frac{1}{\sqrt{n}})$, with high probability the ball walk takes a step of size $\frac{1}{n}$ in the direction of this face.  
 Even though the ball walk may have an $\Omega(1)$ probability of proposing a step which violates the inequality corresponding to the closest face,  if one only checks a small random subset of size $\frac{m}{n}$ of the inequalities one is likely to miss checking the inequality corresponding to this face, which in many cases would cause the Markov chain to leave the polytope with probability $\Omega(1)$ at each step (this is the case, for instance, if the polytope is a cube).  
 Hence, we cannot limit ourselves to checking a random subsample of the inequalities.

\paragraph{Our method of computing polytope membership.} 
  Recall that the isoperimetric inequality \cite{osserman1978isoperimetric}
  implies that any Markov chain whose stationary distribution is close to the uniform distribution on the polytope will spend on average at least $\frac{1}{n}$ of its time a distance $\frac{1}{n}$ from the boundary of the polytope.  
  This suggests that if we wish to compute only a small subsample of the polytope inequalities at each step of the Markov chain, we must make sure that our subsample includes all those inequalities whose corresponding face is close to the current Markov chain step.  
  While it is possible for us to compute the distance $h_j(i) := A_j X_i - b_j$ between the current point in the Markov chain $X_i$ to the hyperplane $H_j$ corresponding to each inequality $(A_j, b_j)$, we do not wish to compute this distance $h_j(i)$ for each $j$ at every step $i$ in the Markov chain, since this takes the same $mn$ operations needed to check each inequality.  
  To get around this problem, rather than computing each distance, we instead only compute the distance to any given hyperplane at a small fraction of the steps.  
  To determine which $h_j(i)$ should be computed at any given step, we estimate a high-probability lower bound $L_j(i)$ for $h_j(i)$ and only compute $h_j(i)$ if the Markov chain is likely to propose a step violating the $j$th inequality, that is, if this lower bound is $L_j(i) = O(\frac{1}{n})$.
  
   \smallskip
 To estimate the lower bound on $h_j(i)$, we apply concentration inequalities to the steps of the ball walk.  
 Specifically, we use concentration inequalities for spherical caps to show that with probability at least $1-\frac{\epsilon}{m}$ the ball walk will never take any steps of size more than $\frac{\eta}{\sqrt{n}} \log(\frac{i_{\mathrm{max}}}{\epsilon})$ in the direction of $H_j$, if the ball walk is run for at most $i_{\mathrm{max}}$ steps ($i_{\mathrm{max}}$ is a ``cutoff time" parameter which can be set by the user of Algorithm \ref{alg:BallModified}; if the algorithm takes more than $i_{\mathrm{max}}$ steps, the algorithm terminates without outputting any samples).   
  Hence, if we set our high-probability lower bound to be $L_j(i) = h_j(\mathbbm{i}^\star(i)) - (i- \mathbbm{i}^\star(i)) \times \frac{\eta}{\sqrt{n}} \log(\frac{i_{\mathrm{max}}}{\epsilon})$,  where $\mathbbm{i}^\star(i)$ is the last time before step $i$ that the distance to $H_j$ was computed, then with probability at least $1-\epsilon$ the ball walk will never leave the polytope (Lemma \ref{thm:failure}, and Step 13 in Algorithm \ref{alg:BallModified}).

\paragraph{Using anti-concentration to prove expected running time bounds for our algorithm.}
Even though we have shown that our algorithm obtains the exact same samples as the usual implementation of the ball walk with high probability, we still have to show that it reduces the number of inequalities one has to check at each step of the Markov chain. 
Towards this end, we prove an anti-concentration inequality (Lemma \ref{assumption:main}) for the uniform distribution on a convex body to show that the expected number of inequalities our algorithm checks at any given step is roughly $\frac{m}{n}$.  
Roughly, this inequality says that a uniform random point on an isotropic convex body has probability at most $O(\Delta)$ of being within a distance $\Delta$ of any given codimension-$1$ hyperplane. 
The main obstacle in applying the anti-concentration inequality is that, while we are able to guarantee that a random step of the Markov chain is $O(1)$-warm with respect to the uniform distribution on $K$, if we are to only check a small fraction $\frac{1}{n}$ of the inequalities (in expectation) at each step, the steps where one checks the inequality cannot be uniformly distributed or even $O(1)$-warm with respect to the uniform distribution but instead can only be $\Theta(n)$-warm at best.  
This is because one has to ensure that, if the Markov chain is within a distance of one ball walk step from a given inequality, this inequality will be checked with very high probability, but only with probability $O(\frac{1}{n})$ at a ``typical" step.  

 \smallskip
To get around this problem we instead bound the waiting time between any given step $i$ of the Markov chain and the most recent time $\mathbbm{i}^\star(i)$ that the distance to the hyperplane was computed by the algorithm, as a function of the distance to this hyperplane at the current step $i$ of the Markov chain, allowing us to apply the anti-concentration inequality.  
Specifically, we use our bound on the size of the ball walk's steps in the direction of $H_j$ to show that that $i-\mathbbm{i}^\star(i)$ must be greater than $\frac{h_j(i)\sqrt{n}}{\eta \log(\frac{m i_{\mathrm{max}}}{\epsilon})}$. Since $\eta = \Theta(\frac{1}{\sqrt{n}})$, applying the anti-concentration inequality on $h_j(i)$ gives a lower bound on $\mathbb{E}[i-\mathbbm{i}^\star(i)]$ that is roughly equal to $n$, if we are given an $O(1)$-warm start with respect to the uniform distribution.  
This in turn implies that given an $O(1)$-warm start, we compute the distance to any given hyperplane at only a small fraction $\frac{1}{n}$ of the ball walk steps in expectation (Lemma \ref{lemma:Frequency}).  
Hence, rather than taking $mn$ operations, our algorithm is able to compute each step of the ball walk in only $m$ operations in expectation, an improvement by a factor of $n$.
 
\paragraph{Improved bounds for Rounding a polytope.} 
Unfortunately, we cannot use the rounding meta-algorithm of \cite{lovasz2006simulated}, which requires only $n\log(n)$ samples to round the polytope, with our implementation of the ball walk.
  The reason is that, while the rounding meta-algorithm in \cite{lovasz2006simulated} keeps a sequence of convex bodies well-rounded at each of the $n$ iterations, it only keeps the convex bodies in isotropic position at a small number $\log(n)$ of the iterations.  
  To obtain our bounds on the expected frequency at which one needs to compute the distance to each $H_j$, we must make sure that the convex body is in isotropic position at each iteration; our rounding algorithm (Algorithm \ref{alg:rounding}) uses $n^2 \log(n)$ samples instead of $n\log(n)$ (see the next paragraph for a discussion of the rounding algorithm).  
  Hence, it would seem at first that, despite the fact that we improve the expected number of arithmetic operations at each step of the Markov chain by a factor of $n$, this improvement would be offset by the fact that we need $n$ times as many samples, which we require to keep the polytope in isotropic position at every iteration.  
  However, there are additional benefits to keeping a polytope in isotropic position.  
  In particular, recent improvements towards weaker versions of the KLS conjecture  imply that the best current bound on the mixing time of the proper steps of the ball walk (also called the ``speedy walk") also improves to $n^{2.5}$ by a factor of $\sqrt{n}$ if the polytope is in isotropic position as opposed to the roughly $n^3$ mixing time bound available for the hit-and-run Markov chain when the polytope is well-rounded but not $O(1)$-isotropic \cite{lee2017eldan}.  
  Hence, keeping the convex body in isotropic position allows us to combine our factor of $n$ improvement with the $\sqrt{n}$ improvement in the mixing time from \cite{lee2017eldan}.  
  The number of arithmetic operations to round a convex body is then at most roughly $mn^{4.5}$, an improvement of $\sqrt{n}$ over the $mn^5$ bound of \cite{lovasz2006simulated}.
   Note that it is enough to bound the number of arithmetic operations in expectation, since one can always start over if the rounding algorithm takes more than its expected number of steps.

\paragraph{Rounding a polytope by sampling from isotropic position (Algorithm \ref{alg:rounding}).} 
More specifically, in order to efficiently generate the samples needed to bring a polytope into $2$-isotropic position one should first ensure that the polytope from which one samples is, say, in $15$-isotropic position.  
Towards this end, one can consider a sequence of nested convex bodies $K_i := K \cap (1+\nicefrac{1}{n})^i r B$.  
The initial polytope $K_0 = rB$ is just the ball contained inside $K$, which can be brought into isotropic position by multiplying this ball by $\frac{1}{r\sqrt{d}}$.  
Since the diameter of $K_i$ increases by a factor of only $1+\nicefrac{1}{n}$, at each step, one can show that the volume of these convex bodies does not increase by more than $e$ at each step $i$, and that for any transformation that brings $K_{i-1}$ into $2$-isotropic position, applying the same transformation to $K_{i}$ would bring it into 15-isotropic position.  
This suggests an iterative algorithm (Algorithm \ref{alg:rounding}), where one samples from a $15$-isotropic convex body which is a linear transformation of $K_{i-1}$, allowing one to bring $K_{i-1}$ into $2$-isotropic position.  
The same transformation brings $K_{i}$ into $15$-isotropic position, allowing one to iteratively bring the sequence of polytopes into $2$-isotropic position by alternating between sampling and linear transformation steps (Lemma \ref{lemma:isotropic}).  
This takes $n\log(\frac{R}{r})$ iterations to bring the polytope $K$ into isotropic position, and uses $n\log(n)$ samples at each iteration;  the number of samples to round the polytope is then roughly $n^2\log(n)\log(\frac{R}{r})$.  
To allow us to apply Lemma \ref{lemma:Frequency} and bound the expected fraction of the time that our implementation of the ball walk (Algorithm \ref{alg:BallModified}) checks a given inequality, we must still show the ball walk has a warm start at each iteration of Algorithm \ref{alg:rounding}.  
We can obtain a warm start for $K_{i}$ by using a sample from the $i-1$ iteration which is approximately uniformly distributed on $K_{i-1}$  (Steps 15-21 of Algorithm \ref{alg:rounding}).  Since $K_{i-1}$ contains at least $\frac{1}{e}$ of the volume of $K_i$, a uniformly distributed point on $K_{i-1}$ provides us with an $O(1)$-warm start for $K_{i}$ (Lemma \ref{lemma:warmness}).  

 \smallskip
\begin{remark}
 If the KLS conjecture is proved true, the mixing time bound of the speedy walk on convex bodies in isotropic position would decrease by an additional factor of $\sqrt{n}$ to just roughly $n^2$, potentially allowing us to improve our running time to roughly $m n^4$.  
 On the other hand, as noted in \cite{lee2017eldan}, it is not known how to connect improvements in the KLS conjecture to the current-best rounding algorithm which uses hit-and-run from a well-rounded but not isotropic position \cite{lovasz2006simulated}.  
 We note, however, that since further improvements to KLS would only apply to the ball walk from a warm start, even using our method, one would have to find a way to modify our rounding algorithm to allow it to use approximately-independent samples from a warm start rather than fully-independent samples from a cold start.
\end{remark}

\paragraph{Organization of the rest of the paper.}
In the rest of the paper, we prove the main result (Theorem \ref{thm:main_intro}, proved for our specific algorithm as Theorem \ref{thm:main_AlgorithmSpecific}), and its corollaries \ref{thm:intro_volume} and \ref{cor:sampling}.

In Section \ref{sec:ProofSamplingFromIsotropicPosition} we bound the accuracy of Algorithm \ref{alg:BallModified} and (roughly speaking) the expected number of arithmetic operations it performs when it is used to sample from a polytope in $O(1)$-isotropic position. 
 In Section \ref{section:failure} we bound the probability our implementation of the ball walk leaves the polytope. 
  In Section \ref{sec:AntiConcentration} we prove an anti-concentration bound, which we use in Section \ref{sec:Frequency} to bound the expected frequency at which our implementation of the ball walk checks a given inequality. 
   In Sections \ref{sec:SpeedyWalkMixing} and \ref{sec:BoundingAccuracy} we recall results from \cite{kannan1997random} and \cite{lee2017eldan} which allow us to then bound the mixing time of the speedy walk (the proper steps of the ball walk) and the expected number of steps.

In Section \ref{section:RoundingProof} we bound the success probability and expected number of arithmetic operations of the rounding Algorithm (Algorithm \ref{alg:rounding}).  In Section \ref{SecRoundingSuccess} we bound the success probability of Algorithm \ref{alg:rounding}.  In Section \ref{sec:ExpectedRunningTimeRounding} we bound the expected number of arithmetic operations made by Algorithm \ref{alg:rounding} under the assumption that it provides a warm start to the ball walk subroutine (Algorithm \ref{alg:BallModified}) at each iteration of Algorithm \ref{alg:rounding}, and in Section \ref{sec:warmness} we show that this warm start assumption holds.  In Section \ref{sec:NegligibleContributions} we verify that the running time of steps where one does not check inequalities have only negligible contribution to the running time of Algorithm \ref{alg:rounding}.  In Section \ref{sec:proofrounding} we combine these results to complete the proof of our main theorem for rounding (Theorem \ref{thm:main_AlgorithmSpecific}).

In Section \ref{sec:CorollaryProofs} we prove our results for volume computation (Corollary \ref{thm:intro_volume}) and sampling (Corollary \ref{cor:sampling}) for polytopes which may be far from isotropic position.

\section{Sampling from an isotropic position.} \label{sec:ProofSamplingFromIsotropicPosition}
\subsection{Bounding the distance traveled in any direction.} \label{section:failure}
In this section we bound the distance traveled by the Markov chain in the direction orthogonal to the plane $H_j$ after $i$ steps.

\begin{lemma} \label{thm:failure}
Fix $\hat{\epsilon}>0$ and suppose that $\alpha \geq 4\log\left(\frac{2  m i_{\mathrm{max}}}{\hat{\epsilon}} \right)$ in Algorithm \ref{alg:BallModified}. 
 Then with probability at least $1-\hat{\epsilon}$ we have that, given the same random vectors $\xi_i$, the output of Algorithm \ref{alg:BallModified} is the same regardless of whether we set $\mathrm{Modified} = \mathrm{ON}$ or $\mathrm{Modified} = \mathrm{OFF}$.
\end{lemma}

\begin{proof}
By the concentration inequality for spherical caps \cite{ledoux2001concentration}, for $\xi \sim \mathrm{uniform}(B(0,1))$, for the $j$th row $A_j$ of the matrix $A$ we have
\be
\mathbb{P}\left(\left| A_j\frac{\xi}{\|\xi\|_2} \right| \geq t\right) \leq 2 e^{-(n-2) t^2/2},
\ee
and hence
\be
\mathbb{P}\left(\left|\eta A_j\frac{\xi}{\|\xi\|_2} \right| \geq \frac{\eta t}{\sqrt{n}}\right) \leq 2 e^{-t^2/2}.
\ee
Thus,
\be
\mathbb{P}\left(\left|\eta A_j\frac{\xi}{\|\xi\|_2} \right| \geq \frac{\eta}{\sqrt{n}} 2 \log(\frac{2}{\delta}) \right) \leq \delta,
\ee
for every $\delta>0$.
Therefore, for every $\hat{\epsilon}>0$,
\be
\mathbb{P}\left(\sum_{\ell=1}^i \left|\eta A_j\frac{\xi_\ell}{\|\xi\|_2} \right| \geq i\times \frac{\eta}{\sqrt{n}} 2\log\left(\frac{2 i}{\hat{\epsilon}}\right)\right) \leq \mathbb{P}\left(\left|\eta A_j\frac{\xi_\ell}{\|\xi\|_2} \right| \geq \frac{\eta}{\sqrt{n}} 2\log\left(\frac{2 i}{\hat{\epsilon}}\right) \textrm{ for some } \ell \in [i]\right) \leq i \times \frac{\hat{\epsilon}}{i}
 = \hat{\epsilon}.
\ee

\noindent
Hence, for our implementation of the ball walk $X_1, X_2,\ldots$ in Algorithm \ref{alg:BallModified} we have
\be
\mathbb{P}&\left(\sup_{i \leq \ell \leq k} |A_jX_{\ell} -A_jX_i| \geq (k-i)\times \frac{\eta}{\sqrt{n}} 2 \log\left(\frac{2i_{\mathrm{max}}}{\hat{\epsilon}}\right) \qquad \textrm{for some } 0\leq i  \leq k \leq i_{\mathrm{max}}\right)\\
&\leq  \mathbb{P}\left(\left|\eta A_j\frac{\xi_\ell}{\|\xi\|_2} \right| \geq \frac{\eta}{\sqrt{n}} 2\log\left(\frac{2 i_{\mathrm{max}}}{\hat{\epsilon}}\right) \textrm{ for some } \ell \in [i_{\mathrm{max}}]\right)\\
&\leq i_{\mathrm{max}} \times \frac{\hat{\epsilon}}{i_{\mathrm{max}}}\\
&\leq \hat{\epsilon}.
\ee

\noindent
Thus, if $\alpha \geq 2\log\left(\frac{2 m  i_{\mathrm{max}}}{\hat{\epsilon}} \right)$, we have
\be \label{eq:failure}
\mathbb{P}\left(\sup_{i \leq \ell \leq k} |A_jX_{\ell} -A_jX_i| \geq (k-i)\times \frac{\eta}{\sqrt{n}} \alpha \qquad \textrm{for some } 0\leq i \leq k \leq i_{\mathrm{max}}\right) \leq \frac{\hat{\epsilon}}{m}.
\ee

\noindent
Let $Z_i$ be the usual ball walk Markov chain (where we check every inequality at each step) which evolves according to the following update equations:
\be
Z_{0} &= X_0\\
Z_{i+1} &= \begin{cases} Z_i + \eta \xi_i \qquad \textrm{ if } Z_i + \eta \xi_i \in K\\ Z_i \qquad \qquad \textrm{ otherwise.}  \end{cases}
\ee
Then inequality \eqref{eq:failure} implies that if we set  $\alpha \geq 2\log\left(\frac{2  m i_{\mathrm{max}}}{\hat{\epsilon}} \right)$ in Algorithm \ref{alg:BallModified}, then with probability at least $1- \hat{\epsilon}$ we have that $X_i = Z_i$ for every $i \in [i_{\mathrm{max}}]$.

Therefore, with probability at least $1- \hat{\epsilon}$, the output of Algorithm \ref{alg:BallModified} is the same regardless of whether we set $\mathrm{Modified} = \mathrm{ON}$ or $\mathrm{Modified} = \mathrm{OFF}$.
\end{proof}

\begin{remark}
We have to bound the sum of the absolute value of the distance, rather than the sum of the variance, since the rejection step could introduce a bias (for instance, if the Markov chain is traveling along the face of a polytope).
\end{remark}

\subsection{Anti-concentration bounds for isotropic convex bodies} \label{sec:AntiConcentration}

\begin{lemma}\label{assumption:main}
Suppose that the uniform distribution $\pi_K$ on $K$ has identity covariance matrix (that is, $K$ is $1$-isotropic) and that $X \sim \pi_K$ is a random vector uniformly distributed on $K$. Let $H$ be any codimension-1 hyperplane.  
 Then we have 
\be
\mathbb{P}(\mathrm{dist}(X, H) \leq \hat{\epsilon}) \leq 2 \hat{\epsilon} \qquad \forall  \hat{\epsilon} >0.
\ee
\end{lemma}

\begin{proof}
Let $\mathsf{A} \in \mathbb{R}^{1 \times n}$ be a row vector and $\mathsf{b} \in \mathbb{R}$ a real number such that $\mathsf{A}x = \mathsf{b}$ is the equation for the Hyperplane $H$. 
 Let $X\sim \pi_K$ be a random vector uniformly distributed on $K$.  Denote the distribution of $\mathsf{A} X$ by $\pi_K^{\mathsf{A}}$.  Note that $\pi_K^{\mathsf{A}}$ is a marginal distribution of $\pi_K$. 
  First, we note two facts:
\begin{enumerate}
\item All marginals of a logconcave distribution are logconcave (Theorem 2.2 of \cite{vempala2005geometric}).
\item If the covariance matrix of any distribution $\pi$ satisfies $\sigma_1 I_n \preccurlyeq \Sigma \preccurlyeq \sigma_2 I_n$, then the variance $\mathsf{A} \Sigma \mathsf{A}^\top$ of its marginal  $\pi^A$ in the subspace defined by $\mathsf{A}$ satisfies $\sigma_1  \leq \mathsf{A} \Sigma \mathsf{A}^\top \leq \sigma_2$. 
\end{enumerate}

\noindent
By the above facts, we have that the distribution of $\mathsf{A}X \in \mathbb{R}$ is isotropic (i.e., it has variance 1) and is logconcave.  Let $x^\star$ be a maximizer of $\pi_K^{\mathsf{A}}$. 
 By Lemma 5.5(a) in  \cite{lovasz2007geometry}, we have
\be \label{eq:KLS_weaker}
\pi_K^{\mathsf{A}}(x^\star) \leq 1,
\ee
 for some universal constant $\hat{c}$.
 
 Hence, for any $\hat{\epsilon}>0$ we have
 \be
 \mathbb{P}(|\mathsf{A}X - \mathsf{b}| \leq \hat{\epsilon}) &= \int_{\mathsf{b}-\hat{\epsilon}}^{\mathsf{b}+\hat{\epsilon}} \pi_K^{\mathsf{A}}(x) \mathrm{d}x\\
 & \leq \int_{\mathsf{b}-\hat{\epsilon}}^{\mathsf{b}+\hat{\epsilon}} \pi_K^{\mathsf{A}}(x^\star) \mathrm{d}x\\
 &  \stackrel{{\scriptsize (\textrm{Eq.  \ref{eq:KLS_weaker}})}}{\leq} \int_{\mathsf{b}-\hat{\epsilon}}^{\mathsf{b}+\hat{\epsilon}} 1 \mathrm{d}x\\
 & = 2\hat{\epsilon}.
 \ee
\end{proof}

\begin{remark}
Note that the bounds in \cite{paouris2012small} are more general than what we need since they apply to hyperplanes of any codimension. 
 We only care about codimension-1 hyperplanes, and can reduce the problem to obtaining anti-concentration bounds for a 1-dimensional isotropic logconcave distribution. 
  This allows us to get a tight bound in Lemma \ref{assumption:main} without assuming the KLS conjecture.  This bound is tight (up to a universal constant) since it is tight for the special case of the unit cube and the regular simplex.
\end{remark}

\subsection{Bounding the frequency of constraint checking} \label{sec:Frequency}

To simplify notation, define $\hat{\eta} := \frac{1}{10}\eta \sqrt{n}$, and $\gamma:=10\alpha \hat{\eta}$.

\begin{lemma} \label{lemma:Frequency}
Suppose that $K$ contains a ball of radius $r=\frac{1}{10}$. 
 Fix $\hat{\epsilon}>0$ and set the algorithmic parameter  $\alpha \geq 4\log\left(\frac{2  i_{\mathrm{max}}}{\hat{\epsilon}} \right)$. 
  Consider any row $A_j$ of A and entry $b_j$ of $b$.  Suppose that the initial point $X_0$ is $\beta$-warm with respect to the uniform distribution for some $\beta>0$. 
   Let $N_j$ be the number of steps (excluding the first step) of the Markov chain in Algorithm \ref{alg:BallModified} with $\mathrm{Modified} = \mathrm{ON}$ at which the algorithm checks inequality $(A_j, b_j)$ and let $N$ be the number of Markov chain steps. 
    Let  $F_j := \frac{N_j}{N}$ be the frequency of checking this inequality (excluding the first check).  Then
\be
\mathbb{E}[F_j] \leq 16n^{-1} \gamma \beta +  \frac{32 \gamma}{n} \times \beta \log\left(\nicefrac{n}{\gamma}\right) + \frac{1}{N}\beta \hat{\epsilon}.
\ee
\end{lemma}
\begin{proof}
First, we note that the stationary distribution of the steps of the ball walk (including improper and proper steps) is uniform on $K$. 
 Let $X= X_0,  X_1, X_2,\ldots X_N$ be the Markov chain generated by Algorithm \ref{alg:BallModified} with $\mathrm{Modified} = \mathrm{ON}$, initial point $X_0 = Y_0$, and random vectors $\xi_1,\ldots$.  
  Let $\tilde{X} = \tilde{X}_1, \tilde{X}_2,\ldots \tilde{X}_N$ be the first $N$ steps of the Markov chain generated by Algorithm \ref{alg:BallModified} with $\mathrm{Modified} = \mathrm{OFF}$, using the same initial point $\tilde{X}_0 = Y_0$ and the same random vectors \footnote{We consider all steps $\tilde{X}_k = \tilde{X}_{\tilde{N}}$, for all $k \geq \tilde{N}$,where $\tilde{N}$ is the number of Markov chain steps computed by the algorithm.  That is, the Markov chain remains stuck forever at the same point after Algorithm \ref{alg:BallModified} halts.}. 
   Using the same initial point and random vectors defines a coupling between $X$ and $\tilde{X}$. Then $X_k = \tilde{X}_k$ for all $k$ if and only if $X_k \in K$ for all $k \leq N$.

Let $G$ be the event that $\sup_{i \leq \ell \leq k} |A_j\tilde{X}_{\ell} -A_j\tilde{X}_i| < (k-i)\times \frac{\eta}{\sqrt{n}} \alpha$ for all $0\leq i \leq k  \leq i_{\mathrm{max}}$ and all $j \in [m]$. 
  By Equation \eqref{eq:failure} in the proof of Lemma \ref{thm:failure}, we have that 
  \be
  \mathbb{P}(G) \geq 1- \beta\hat{\epsilon}.
  \ee 
   Also by the proof of Lemma \ref{thm:failure} we have that $X_k \in K$ for all $k \leq N$ if $G$ occurs.  Hence,
\be \label{eq:ExitProbability}
\mathbb{P}(X_k = \tilde{X}_k \forall k \leq  i_{\mathrm{max}}) \geq \mathbb{P}(G).
\ee

\noindent
Recall that $h_j(i) := b_j - A_j X_i$ is the distance from the Markov chain $X_i$ to the hyperplane corresponding to the inequality $(A_j, b_j)$ at step $i$. 
 Rather than checking the inequality $(A_j, b_j)$ at each step of the ball walk, Algorithm \ref{alg:BallModified} waits some number of steps $w(i)$ after checking this inequality at some step $i$.  
  More generally, we define 
  \be
  w(i) := \max\left(\left\lfloor \frac{\sqrt{n}}{\alpha \eta} h_j(i) \right\rfloor,1\right)
  \ee
   regardless of whether the inequality is actually checked at step $i$ (we can think of $w(i)$ as the amount of time the algorithm \textit{would} have waited if it had checked the inequality at step $i$).

 Let $\mathbbm{i}(k)$ be the step at which the inequality is checked for the $k$th time. 
  Let $\mathbbm{k}(i)$ be the number of times the inequality has been checked after $i$ Markov chain steps (in particular, we have $\mathbbm{k}(i) \leq i$). 
   Let $\mathbbm{i}^\star(i) := \mathbbm{i}( \mathbbm{k}(i))$ be the last time the inequality was checked.

  Then, if the Markov chain $X$ does not leave $K$, the total number $N_j$ of times the inequality is checked is
\be \label{eq:Frequency0}
N_j = \sum_{k =1}^{N_j} w(\mathbbm{i}(k)) = \sum_{i=1}^N \frac{1}{w(\mathbbm{i}^{\star}(i))} \qquad \qquad \textrm{(If $G$ does not occur),}
\ee
where the first equality holds because $w(\mathbbm{i}(k))=1$ for all $k$.
Therefore, we have
\be \label{eq:Frequency1}
\mathbb{E}[N_j] \leq \sum_{i=1}^N \mathbb{E}\left[\frac{1}{w(\mathbbm{i}^{\star}(i))}\right] + \mathbb{P}(G^c).
\ee

\noindent
We will show that, if $G$ occurs, then $w(i-s) \geq \frac{1}{4} w(i)$ for all $0\leq s\leq \frac{1}{8}w(i)$.

Suppose that $G$ occurs. 
 Without loss of generality we may assume that $w(i) >2$ (since otherwise we have $w(i-s) > 1 > \frac{1}{4} w(i)$). 
  Then
\be \label{eq:waiting1}
h_j(i-s) &\geq h_j(i) - s \times \frac{\eta}{\sqrt{n}} \alpha\\
&\geq h_j(i) - \frac{1}{8}w(i) \times \frac{\eta}{\sqrt{n}} \alpha\\
&=  h_j(i) - \frac{2}{8}\frac{\sqrt{n}}{\alpha \eta}h_j(i) \times \frac{\eta}{\sqrt{n}} \alpha\\
&= \frac{3}{4} h_j(i).
\ee

\noindent
Therefore,
\be
w(i-s) &= \max\left(\left \lfloor \frac{\sqrt{n}}{\alpha \eta} h_j(i-s) \right\rfloor, \, \,1\right)\\
&\stackrel{{\scriptsize (\textrm{Eq.  \ref{eq:waiting1}})}}{\geq} \max\left(\left\lfloor \frac{\sqrt{n}}{\alpha \eta} \frac{3}{4}h_j(i) \right \rfloor,  \, \, 1\right)\\
&\geq \max\left( \frac{\sqrt{n}}{\alpha \eta} \frac{3}{4}h_j(i), \, \,1\right) -1\\
&\geq \frac{3}{4} \max\left( \frac{\sqrt{n}}{\alpha \eta} h_j(i),  \, \, 1\right) -1\\
&\geq \frac{3}{4}w(i) -1\\
&\geq \frac{1}{4}w(i),
\ee
where the last inequality holds since we assumed without loss of generality that $w(i) >2$. 
 Therefore whenever $G$ occurs we have
\be  \label{eq:waiting2}
w(i-s) \geq \frac{1}{4} w(i) \qquad \qquad \forall i,s \in \mathbb{Z}^+,  \,\, s\in\left[0,\frac{1}{8}w(i)\right].
\ee
Suppose (towards a contradiction) that $w(\mathbbm{i}^{\star}(i)) < \frac{1}{8} w(i)$. 
 But we always have $\mathbbm{i}^{\star}(i) +w(\mathbbm{i}^{\star}(i)) >i$ (since $w(\mathbbm{i}^{\star}(i))$ is the amount of time we wait to check the inequality after step $\mathbbm{i}^{\star}(i)$, and, by definition of $\mathbbm{i}^{\star}(i)$ we have not yet re-checked the inequality at step $i$). 
  Hence we would have $i - \mathbbm{i}^{\star}(i)< w(\mathbbm{i}^{\star}(i))<\frac{1}{8} w(i)$. 
   Then by Inequality \ref{eq:waiting2} we would have
\be
w(\mathbbm{i}^{\star}(i)) \geq \frac{1}{4} w(i),
\ee
 which contradicts our assumption that $w(\mathbbm{i}^{\star}(i)) < \frac{1}{8} w(i)$.  Therefore, by contradiction we have that 
\be \label{eq:waiting3}
w(\mathbbm{i}^{\star}(i)) \geq \frac{1}{8} w(i) \qquad \forall i \in \mathbb{Z}^+.
\ee
Hence, combing Equations \eqref{eq:Frequency1} and \ref{eq:waiting3} we have
\be \label{eq:Frequency2}
\mathbb{E}[N_j] \stackrel{{\scriptsize (\textrm{Eq.  \ref{eq:Frequency1}})}}{\leq}  \sum_{i=1}^N \mathbb{E}\left[\frac{1}{w(\mathbbm{i}^{\star}(i))}\right]  + \mathbb{P}(G^c)\stackrel{{\scriptsize (\textrm{Eq.  \ref{eq:waiting3}})}}{\leq}  8 \sum_{i=1}^N \mathbb{E}\left[\frac{1}{w(i)}\right]  + \mathbb{P}(G^c).
\ee
Therefore it is enough to bound $\mathbb{E}[\frac{1}{w(i)}]$ for each $i$.

\paragraph{Bounding $\mathbb{E}[\frac{1}{w(i)}]$.}
Fix any $i \in [N]$.
First, we note that without loss of generality we may assume that $X_0 \equiv \tilde{X}_0$ is a $1$-warm start, since the bound on $\mathbb{E}[F_j]$ for the $\beta$-warm case for general $\beta \geq 1$ will be at most $\beta$ times as large as the bound for the $1$-warm special case.

In the special case where $X_0 \equiv \tilde{X}_0$ is a $1$-warm start, $\tilde{X}_i \sim \pi_K$ is uniformly distributed on $K$.
Then by Lemma \ref{assumption:main} we have 
\be \label{eq:waiting4}
\mathbb{E}\left[\frac{1}{w(i)}\right] &= \mathbb{E}\left[\frac{1}{\max(\lfloor \frac{\sqrt{n}}{\alpha \eta} h_j(i) \rfloor,1)}\right]\\
&=\mathbb{E}\left[\frac{1}{\max(\lfloor \frac{n}{10\alpha \hat{\eta}} h_j(i) \rfloor,1)}\right]\\
&\leq 1\times \mathbb{P}\left(h_j(i) \leq \frac{10\alpha \hat{\eta}}{n}\right) + 2\mathbb{E}\left[\frac{10 \alpha \hat{\eta}}{n h_j(i) }  \times \mathbbm{1}\{h_j(i) \geq \frac{10\alpha \hat{\eta}}{n}\}\right]\\
  &\stackrel{{\scriptsize (\textrm{Lemma } \ref{assumption:main})}}{\leq} 2 n^{-1} \gamma + 2\mathbb{E}\left[\frac{\gamma}{n h_j(i) } \times \mathbbm{1}\{h_j(i) \geq \frac{\gamma}{n}\}\right],
\ee
 where $\gamma:=10\alpha \hat{\eta}$.
 
 But
 \be \label{eq:waiting5}
 \mathbb{E}\left[\frac{\gamma}{n h_j(i)} \times \mathbbm{1}\{h_j(i) \geq \frac{\gamma}{n}\}\right] &= \int_0^\infty \mathbb{P}\left(\frac{\gamma}{n h_j(i)} \times \mathbbm{1}\{h_j(i) \geq \frac{\gamma}{n}\} \geq t\right) \mathrm{d}t\\
 &= \int_0^\infty \mathbb{P}\left(\frac{\gamma}{n h_j(i)} \times \mathbbm{1}\{1 \geq \frac{\gamma}{n h_j(i)}\} \geq t\right) \mathrm{d}t\\
 &= \int_0^1 \mathbb{P}\left(\frac{\gamma}{n h_j(i)}  \times \mathbbm{1}\{1 \geq \frac{\gamma}{n h_j(i)}\} \geq t\right) \mathrm{d}t\\
 &\leq \int_0^1 \mathbb{P}\left(\frac{\gamma}{n h_j(i)}  \geq t\right) \mathrm{d}t\\
 &= \int_0^1 \mathbb{P}\left(\frac{\gamma}{t n}  \geq h_j(i)\right) \mathrm{d}t\\
   &= -\frac{\gamma}{n} \int_{\infty}^{\frac{\gamma}{n}} \mathbb{P}\left(u \geq h_j(i)\right) u^{-2} \mathrm{d}u\\
  &= \frac{\gamma}{n} \int_{\frac{\gamma}{n}}^{\infty} \mathbb{P}\left(u \geq h_j(i)\right) u^{-2} \mathrm{d}u\\
  &= \frac{\gamma}{n} \int_{\frac{\gamma}{n}}^{1} \mathbb{P}(u \geq h_j(i)) u^{-2} \mathrm{d}u + \frac{\gamma}{n} \int_{1}^{\infty} \mathbb{P}(u \geq h_j(i)) u^{-2} \mathrm{d}u\\
  &\stackrel{{\scriptsize (\textrm{Lemma } \ref{assumption:main})}}{\leq}\frac{\gamma}{n} \int_{\frac{\gamma}{n}}^{1} u^{1} \times 2 \times u^{-2} \mathrm{d}u + \frac{\gamma}{n} \int_{1}^{\infty} 1\times u^{-2} \mathrm{d}u\\
  &=\frac{2 \gamma}{n} \int_{\frac{\gamma}{n}}^{1} u^{-2+1} \mathrm{d}u + \frac{\gamma}{n}\\
  &=\frac{2 \gamma}{n} \times \log\left(\nicefrac{n}{\gamma}\right)  + \frac{\gamma}{n}\\
    &\leq \frac{4 \gamma}{n} \times \log\left(\nicefrac{n}{\gamma}\right).\\
 \ee
Hence, combining Inequalities \eqref{eq:waiting4} and \eqref{eq:waiting5} we have
\be \label{eq:waiting6}
\mathbb{E}\left[\frac{1}{w(i)}\right] \leq 2 n^{-1} \gamma + \frac{4 \gamma}{n} \times \log\left(\nicefrac{n}{\gamma}\right).
\ee
Thus, by Inequality \eqref{eq:Frequency2} we have
\be \label{eq:Frequency3}
\mathbb{E}[F_j] &= \mathbb{E}\left[\frac{N_j}{N}\right]\\
 &\stackrel{{\scriptsize (\textrm{Eq.  \ref{eq:Frequency2}})}}{\leq}  8 \frac{1}{N}\sum_{i=1}^N \mathbb{E}\left[\frac{1}{w(i)}\right]  + \frac{1}{N}\mathbb{P}(G^c)\\
&\stackrel{{\scriptsize (\textrm{Eq.  \ref{eq:waiting6}})}}{\leq}  \frac{8}{N} \sum_{i=1}^N\left[2n^{-1} \gamma + \frac{4 \gamma}{n} \times \log\left(\nicefrac{n}{\gamma}\right)\right]  + \frac{1}{N}\mathbb{P}(G^c)\\
&=  \frac{8}{N} \times N\left[2 n^{-1} \gamma + \frac{4 \gamma}{n} \times \log\left(\nicefrac{n}{\gamma}\right)\right]  + \frac{1}{N}\mathbb{P}(G^c)\\
&\leq  16 n^{-1} \gamma + \frac{32 \gamma}{n} \times \log\left(\nicefrac{n}{\gamma}\right)  + \frac{1}{N}\hat{\epsilon}.\\
\ee
Hence, in the general-$\beta$ case we get:
\be
\mathbb{E}[F_j] \leq 16 n^{-1} \gamma \beta +  \frac{32 \gamma}{n} \times \beta \log\left(\nicefrac{n}{\gamma}\right) + \frac{1}{N}\beta \hat{\epsilon}.
\ee
\end{proof}

\subsection{Mixing time of the speedy walk}\label{sec:SpeedyWalkMixing}

To bound the mixing time of the ball walk, one can consider the speedy walk. 
 The speedy walk is the same Markov chain as the ball walk except that we leave out the steps where the ball walk does not change position. 
  Since the ball walk ends up staying for more time at certain points than the speedy walk, the speedy walk has a different stationary distribution $\hat{\pi}_K$ called the ``speedy distribution".  
   Denote by the random variable $\tau_i$ the stopping time which is equal to the number of proper+improper steps taken until the ball walk has taken $i$ proper steps.  Then the random walk $Z_1, Z_2, \ldots$ where $Z_i = X_{\tau_i}$ is the ``speedy walk".

We recall the following Theorem\footnote{The ``M-distance" used in \cite{kannan1997random} is bounded above by the warmness $\beta$, and bounded below by the TV distance.  So the result we quote here is in fact weaker than the ``M-distance" version of the result.}
 from \cite{kannan1997random}, of which Theorem 18 and the following paragraph in \cite{lee2017eldan} is a corollary:
\begin{lemma} [Speedy walk (Theorem 18 and following paragraph in \cite{lee2017eldan}, Theorem 4.1 in \cite{kannan1997random})] \label{lemma:SpeedyWalk}
Suppose that $K$ is 15-isotropic and fix $\hat{\epsilon}>0$. 
 Given an initial point $X_0$ which is a $\beta$-warm start with respect to the speedy distribution, the ball walk on $K$ with step size $\eta \geq \frac{1}{800\sqrt{n\log(\nicefrac{n}{\hat{\epsilon}})}}$ satisfies
\be
\|\mathcal{L}(X_{\tau_i}) - \hat{\pi}_K\|_{\mathrm{TV}} \leq \hat{\epsilon}
\ee
if $i \geq c n^{2.5} \log^3(\frac{\beta}{\hat{\epsilon}})$ where $c>0$ is a universal constant.

If instead the ball walk starts from a non-random point which is a distance at least $n^{-c_1}$ for any constant $c_1$, then the ball walk on $K$ with step size  $\eta \geq \frac{1}{800\sqrt{n\log(\nicefrac{n}{\hat{\epsilon}})}}$ satisfies 
\be
\|\mathcal{L}(X_{\tau_i}) - \hat{\pi}_K\|_{\mathrm{TV}} \leq \hat{\epsilon}
\ee
if $i \geq c_2 n^{2} D (\log \log D) \log^3(\nicefrac{n}{\hat{\epsilon}})$ where $D$ is the diameter of $K$ and $c_2$ is a constant that depends only on $c_1$.
\end{lemma}
\noindent Let $\lambda$ be the probability that the ball walk proposes a step inside the convex body $K$ from a point uniformly distributed on $K$; we call $\lambda$ the \textit{average local conductance}. 
 We will use the following results \cite{kannan1997random} which allow one to obtain improved average-case running time bounds for the ball walk.\footnote{Since our goal is to put the convex body in isotropic position, which fails with exponentially small probability in the running time if we obtain iid points, it is enough to bound the average-case running time since we can always just start over if the running time ends up being too long.}
\begin{lemma}[Corollary 4.6 in \cite{kannan1997random}] \label{lemma:AverageConductance}
The average conductance of $K$ satisfies  $\lambda \geq 1- \frac{\eta \sqrt{n}}{2r}$ if $K$ contains a ball of radius $r$.
\end{lemma}
\noindent Using average local conductance  \cite{kannan1997random} gives the following Lemma on the expected number of improper steps taken by the ball walk:
\begin{lemma}[Theorem 4.10b in \cite{kannan1997random}] \label{lemma:NonproperSteps}
Suppose that $X_0$ is distributed according to the speedy distribution on $K$. 
 Fix $t>0$. 
  Then the expected number of (proper and improper) ball walk steps needed to get $t$ proper steps is at most $\frac{2t}{\lambda}$.
\end{lemma}

\begin{remark}
Theorem 4.10b \cite{kannan1997random} was stated for a specific value of $t$ (their bound on the mixing time). 
 However, in the special case when we already start at the stationary distribution of the speedy walk, if the expectation holds for one value of $t$, it must also hold for every value of $t$. 
  (Moreover, we note that even for non-stationary starts (which we do not need here) their \textit{proof} holds for all values of $t$.)
\end{remark}

\subsection{Bounding the accuracy} \label{sec:BoundingAccuracy}

\begin{lemma} \label{lemma:CombinedHitRunSpeedy}
Assume that $K$ is a 30-isotropic convex body containing $B(0,\frac{1}{10})$, and that $\mathrm{dist}(X_0, \partial K) \geq n^{-3}$. 
 Fix $\hat{\epsilon} >0$. 
  Then Algorithm \ref{alg:BallModified} with $\eta \leq \frac{1}{10 \sqrt{8n \log(\nicefrac{n}{\hat{\epsilon}})}}$, $\mathrm{Modified} = \mathrm{OFF}$, and $\mathcal{I} = c_2 n^{2} \rho (\log \log \rho) \log^3(\nicefrac{n}{\hat{\epsilon}})$ outputs independent samples $Z_1, Z_2,\ldots, Z_p$, where each $Z_i$ has TV distance $10\hat{\epsilon}$ to the uniform distribution on $K$.
\end{lemma}

\begin{proof}
Since the speedy walk is initialized at a point $X_0$, which is a distance at least $n^{-3}$ from the boundary of $K$, by Lemma \ref{lemma:SpeedyWalk} we have that the samples $Y_1, Y_2, \ldots$ obtained by running the speedy walk for $\mathcal{I} = c_2 n^{2} \rho (\log \log \rho) \log^3(\nicefrac{n}{\hat{\epsilon}})$ proper steps each satisfy $\|\mathcal{L}(Y_s) - \hat{\pi}_K\|_{\mathrm{TV}} \leq \hat{\epsilon}$ for all $s$. 
 Moreover, these points are independent since each run of the speedy walk starts at the same point $X_0$.

By Theorem 4.16 in \cite{kannan1997random} we have that the samples $Z_1, Z_2, \ldots$ obtained from $Y_1, Y_2, \ldots$ by the rejection sampling step in Algorithm \ref{alg:BallModified} are uniformly distributed on $K$ with TV error $10 \hat{\epsilon}$. 
 Moreover, since $Y_1, Y_2, \ldots$ are jointly independent, $Z_1, Z_2, \ldots, Z_p$ are also jointly independent.
\end{proof}

\section{Rounding a polytope}\label{section:RoundingProof}
In this section we analyze the running time and accuracy of Algorithm \ref{alg:rounding}.

\subsection{Bounding the success probability of Algorithm \ref{alg:rounding}}\label{SecRoundingSuccess}
In this section we bound the success probability of Algorithm \ref{alg:rounding}. 
 We use the following lemma (Corollary 11 in \cite{bertsimas2002solving}), which is a corollary of the main result in \cite{rudelson1999random}.

\begin{lemma}[Corollary 11 in \cite{bertsimas2002solving}] \label{lemma:rounding1}
Let $K$ be a convex set. 
 Let $Y_1,\ldots, Y_p$ be iid uniform random points in $K$ and fix $\hat{\epsilon} >0$. 
  Let $\bar{Y}:= \frac{1}{p} \sum_{i=1}^p Y_k$ and let $\hat{\Sigma}_Y := \frac{1}{p} \sum_{i=1}^p(Y_i-\bar{Y})(Y_i-\bar{Y})^\top$. 
   Then there exists an absolute constant $c$ such that if $p\geq n \times c \log^2(\frac{1}{\hat{\epsilon}}) \log^2(n)$, the convex set $K^\ddagger:=\hat{\Sigma}_Y^{-\frac{1}{2}}(K-\bar{Y})$ is 2-isotropic and $\|\mu_{K^\ddagger}\|_2 < \frac{1}{20}$ with probability at least $1-\hat{\epsilon}$.
\end{lemma}
\noindent Note that in the proof of Corollary 11 in \cite{bertsimas2002solving} it is shown that $\|\mu_{K^\ddagger}\|_2 < \frac{1}{20}$, although this is not mentioned explicitly in \cite{bertsimas2002solving} in the statement of their Corollary.

Fix $\epsilon>0$. From now on we fix the parameters $p$, $\mathcal{I}$, $\eta$ in Algorithm \ref{alg:BallModified} as follows:
\begin{itemize}
\item $p\geq n \times c \log^2(\frac{1}{\epsilon}) \log^2(n)$,
\item $\mathcal{I} = c_2 n^{2} 20 \sqrt{n}\log\left(\frac{40n^2 p^2}{\epsilon}\right) \left(\log \log 20 \sqrt{n}\log\left(\frac{40n^2 p^2}{\epsilon}\right)\right) \log^3\left(\frac{np^2}{\epsilon}\right) \log \log\left(\frac{R}{r}\right)$,
\item $\eta = \frac{1}{30\sqrt{n\log(\nicefrac{n}{\epsilon})}}$.
\end{itemize}

\begin{lemma} \label{lemma:isotropic}
Suppose that we set parameters $\mathrm{modified} = \mathrm{OFF}$, $i_\mathrm{max} = \infty$. 
 Then for any value of $\alpha>0$, with probability at least $1-\epsilon$ the convex body  $\tilde{K}_{i^\star} := \hat{\Sigma}_{i^\star}^{-\frac{1}{2}}(K-  \hat{\mu}_{i^\star})$ outputed by Algorithm \ref{alg:rounding} is in 2-isotropic position. 
  Moreover, the expected number of iterations of each of the ``While" loops in  Algorithm \ref{alg:rounding} is bounded above by $2$.
\end{lemma}

\begin{proof}
Recall the definitions from Algorithm \ref{alg:rounding} where
\be
\tilde{K}_{i-1}:= \begin{cases} B &\textrm{ for }i=1\\
 \hat{\Sigma}_{i}^{-\frac{1}{2}}(K_{i-1}-  \hat{\mu}_{i})&\textrm{ for } i \geq 2,
 \end{cases} 
\ee
and $\hat{K}_i := \hat{\Sigma}_i^{-\frac{1}{2}} (K_i - \hat{\mu}_i)$ for all $i\in \mathbb{N}$.

We prove this theorem by induction:

\textbf{Inductive assumption:} Suppose that $\tilde{K}_{i-1} = \hat{\Sigma}_{i}^{-\frac{1}{2}}(K_{i-1}-  \hat{\mu}_{i})$  is 2-isotropic with $\|\mu_{\tilde{K}_{i-1}}\|_2 \leq \frac{1}{20}$,  and that $B(0, \frac{1}{4}) \subseteq \tilde{K}_{i-1}$.

\textbf{Base case:} 
Since for $i=1$ $\tilde{K}_{i-1}=B$ is 2-isotropic, we must have that $\hat{K}_i$ is $4 e \leq 15$-isotropic (see the inductive case for why this is true). 
 Therefore $\hat{K}_i$ is 15-isotropic and contains the ball $B(0, \frac{1}{4})$ (since it contains $\hat{K}_{i} \supseteq \tilde{K}_{i-1} = B \supseteq B(0, \frac{1}{4})$).

Then by Lemma \ref{lemma:CombinedHitRunSpeedy} we have that the points $Z_1,\ldots, Z_p$ are independent, and are each a TV distance at most $\frac{\epsilon}{p^2}$ from the uniform distribution on $\hat{K}_i$. \footnote{See the inductive case for why we get a bound for the uniform distribution on $\hat{K}_i$ even though the Markov chain is on $\hat{K}_i \cap \rho B$ for $\rho= 20 \sqrt{n}\log(\frac{40 n^2}{\epsilon})B$.}

Therefore, by Lemma \ref{lemma:rounding1} we have that $\tilde{K}_i$ is 2-isotropic with probability at least $1- \frac{\epsilon}{p}$ (since we can couple $Z_1,\ldots, Z_p$ to independent random vectors $\hat{Z}_1,\ldots, \hat{Z}_p$ which are exactly uniformly distributed on $\hat{K}_i$, such that $\mathbb{P}(X_i = \hat{Z}_i) \geq 1-\frac{\epsilon}{p^2}$, and apply Lemma \ref{lemma:rounding1} to these vectors).

\textbf{Inductive case: Showing that $\tilde{K}_{i-1}$ being in 2-isotropic position implies that $\tilde{K}_{i}$ is in 2-isotropic position.}
Since $K_{i-1} \subseteq K_{i}$ and  $\mathrm{Vol}(K_{i}) \leq e\mathrm{Vol}(K_{i-1})$ we have $\Sigma_{K_{i}} \succcurlyeq \frac{\mathrm{Vol}(K_{i-1})}{\mathrm{Vol}(K_{i})}\Sigma_{K_{i-1}} \succcurlyeq \frac{1}{e}\Sigma_{K_{i-1}}$.

Moreover, since $K$ is convex and $0 \in K$, we have 
\be
K_{i} \subseteq (1+\nicefrac{1}{n})K_{i-1} \qquad \textrm{ and  } \qquad \mathrm{Vol}(K_{i}) \geq \frac{1}{e}\mathrm{Vol}\left((1+\nicefrac{1}{n})K_{i-1}\right).
\ee 
 Hence, we have 
 \be
 \Sigma_{K_{i}} \preccurlyeq \frac{\mathrm{Vol}((1+\nicefrac{1}{n})K_{i-1})}{\mathrm{Vol}(K_{i})} \Sigma_{(1+\nicefrac{1}{n})K_{i-1}} = \frac{\mathrm{Vol}((1+\nicefrac{1}{n})K_{i-1})}{\mathrm{Vol}(K_{i})} (1+\nicefrac{1}{n})^2\Sigma_{K_{i-1}} \preccurlyeq 4e \Sigma_{K_{i-1}}.
 \ee

\noindent Therefore, we have that
\be
\frac{1}{e}\Sigma_{K_{i-1}} \preccurlyeq \Sigma_{K_{i}} \preccurlyeq 4e \Sigma_{K_{i-1}},
\ee
and hence that
\be
\frac{1}{e} \frac{1}{4}I_n \preccurlyeq \frac{1}{e}\Sigma_{\tilde{K}_{i-1}} \preccurlyeq \Sigma_{\hat{K}_{i}} \preccurlyeq 2e \Sigma_{\tilde{K}_{i-1}} \preccurlyeq 16e I_n.
\ee
Therefore we have shown that the fact that $\tilde{K}_{i-1}$ is 2-isotropic implies that $\hat{K}_i$ is $4\sqrt{e} \leq 15$-isotropic. 
 Therefore $\hat{K}_i$ is 15-isotropic and contains the ball $B(0, \frac{1}{4})$ (since it contains $\hat{K}_{i} \supseteq \tilde{K}_{i-1} \supseteq B(0, \frac{1}{4})$).

We now show that the centers of mass of $\hat{K}_i$ and $\tilde{K}_{i-1}$ are a distance at most roughly $\sqrt{n}$ apart. Suppose (towards a contradiction) that $\|\mu_{\hat{K}_i} -\mu_{\tilde{K}_{i-1}}\|_2> 10\sqrt{n} \log(\frac{40n}{\epsilon})$. 
 By Lemma 24 in \cite{lee2017eldan}, $1-\frac{\epsilon}{40n}$ of the volume of the convex body $\hat{K}_i$ is inside the ball of radius $2\sqrt{n} \log(\frac{40n}{\epsilon})$ with center at $\mu_{\hat{K}_i}$. 
  Hence, if the assumption $\|\mu_{\hat{K}_i} -\mu_{\tilde{K}_{i-1}}\|_2> 10\sqrt{n} \log(\frac{40n}{\epsilon})$ were true, we would have (for $\epsilon < 0.1$) that a nonzero portion of the volume of $\hat{K}_i$ is a distance of at least $40 n \sqrt{n}$ from $\mu_{\hat{K}_i}$ (since $\frac{\mathrm{Vol}(\hat{K}_i)}{\mathrm{Vol}(\tilde{K}_{i-1})} \leq e$ and $\hat{K}_{i} \supseteq \tilde{K}_{i-1}$). 
   This is a contradiction since the convex body $\hat{K}_i$ is entirely contained in a ball of radius $15n$ because it is $15$-isotropic.  Hence by contradiction we have that 
   \be
   \|\mu_{\hat{K}_i} -\mu_{\tilde{K}_{i-1}}\|_2\leq 10\sqrt{n} \log \left(\frac{40n}{\epsilon}\right).
   \ee

\noindent By inductive assumption we have that $\|\mu_{\tilde{K}_{i-1}}\|_2 \leq \frac{1}{5}$, and hence that 
\be
\|\mu_{\hat{K}_i}\|_2 \leq \|\mu_{\hat{K}_i} -\mu_{\tilde{K}_{i-1}}\|_2 + \|\mu_{\tilde{K}_{i-1}}\|_2 \leq 12\sqrt{n} \log\left(\frac{40n}{\epsilon}\right).
\ee
  By Lemma 24 in \cite{lee2017eldan}, $1-\frac{\epsilon}{40n^2 p^2}$ of the volume of the convex body $\hat{K}_{i}$ is contained in a ball of radius  $2\sqrt{n} \log(\frac{40n^2 p^2}{\epsilon})$ centered at $\mu_{\hat{K}_i}$.  Hence, since $\hat{K}_{i}$ being 15-isotropic implies that it is contained in a ball of radius $15n$, we have that $\hat{K}_{i}^\dagger := \hat{K}_{i} \cap 20 \sqrt{n}\log(\frac{40n^2 p^2}{\epsilon})B$ is 30-isotropic.  Moreover, by Lemma 24 in \cite{lee2017eldan} and the fact that $\|\mu_{\hat{K}_i}\|_2 \leq 12\sqrt{n} \log(\frac{40n}{\epsilon})$, we have 
  $$
  \mathrm{Vol}(\hat{K}_{i}^\dagger) \geq \left(1-\frac{\epsilon}{2 p^2}\right)\mathrm{Vol}(\hat{K}_{i}).
  $$
    Therefore, since the rejection step in the ``While" loop of Algorithm \ref{alg:rounding} ensures that $X_0$ is in the $n^{-3}$-interior of $\tilde{K}_i \cap 20 \sqrt{n}\log(\frac{40 n^2}{\epsilon})B$, by Lemma \ref{lemma:CombinedHitRunSpeedy} we have that the points $Z_1,\ldots, Z_p$ are independent, and are each a TV distance at most $\frac{\epsilon}{p^2}$ from the uniform distribution on $\hat{K}_i$.

Therefore, by Lemma \ref{lemma:rounding1} we have that $\tilde{K}_i$ is 2-isotropic and $\|\mu_{\tilde{K}_i}\|_2 \leq \frac{1}{20}$ with probability at least $1- \frac{\epsilon}{p}$ (Since we can couple $Z_1,\ldots, Z_p$ to independent random vectors $\hat{Z}_1,\ldots, \hat{Z}_p$ which are exactly uniformly distributed on $\hat{K}_i$, such that $\mathbb{P}(Z_k = \hat{Z}_k) \geq 1-\frac{\epsilon}{p^2}$, and apply Lemma \ref{lemma:rounding1} to these vectors.  Therefore, we have $Z_1,\ldots, Z_p$ independent with $\|\mathcal{L}(Z_k) - \pi_{\tilde{K}_{i}}\|_{\mathrm{TV}} \leq \frac{\epsilon}{10p^2 \log(\frac{R}{r})}$ for all $k \in [p]$).

\textbf{Bounding the number of iterations of the ``While" loop"}
First, we bound $\mathrm{dist}(\hat{X}_0', \partial [\tilde{K}_i \cap \rho B])$ for $\rho = 20 \sqrt{n}\log(\frac{40 n^2}{\epsilon})$. 
 Since $\tilde{K}_{i}$ is 2-isotropic, it contains a ball of radius $\frac{1}{\sqrt{2}}$. 
  Therefore, by a statement in the proof of Corollary 4.6 in \cite{kannan1997random}, we have that 
  \be
  \mathrm{Vol}_{n-1}(\partial \tilde{K}_{i}) \leq \frac{n}{\nicefrac{1}{\sqrt{2}}} \mathrm{Vol}(\partial \tilde{K}_{i}).
  \ee 
   Therefore, we have that 
   \be
   \mathbb{P}_{Y \sim \mathrm{unif}(\tilde{K}_{i})}(\mathrm{dist}(Y, \partial \tilde{K}_{i}) < n^{-3}) \leq \frac{n^{-3}\mathrm{Vol}_{n-1}(\partial \tilde{K}_{i})}{\mathrm{Vol}(\partial \tilde{K}_{i})} \leq n^{-2} \sqrt{2}.
   \ee

\noindent Moreover, by Lemma 24 in \cite{lee2017eldan}, $1-\frac{\epsilon}{10n^2 p^2}$ of the volume of the convex body $\tilde{K}_{i}$ is contained in $B(\mu_{\tilde{K}_{i}}, \frac{\rho}{2})$, and since $\tilde{K}_{i}$ is in 2-isotropic position, $\|\mu_{\tilde{K}_{i}}\|_2 \leq \frac{1}{5}$. 
 Thus, $1-\frac{\epsilon}{10n^2 p^2}$ of the volume of the convex body $\tilde{K}_{i}$ is contained in the ball $B(0, \rho)$.

Now, since $\hat{X}_0$ and $Z_1$ are generated by the ball walk with the same starting point and parameters, we have 
\be
\|\mathcal{L}(\hat{X}_0') - \pi_{\tilde{K}_{i}}\|_{\mathrm{TV}}  = \|\mathcal{L}(\hat{X}_0) - \pi_{\hat{K}_{i}}\|_{\mathrm{TV}} = \|\mathcal{L}(Z_1) - \pi_{\hat{K}_{i}}\|_{\mathrm{TV}} \leq \frac{\epsilon}{p^2}.
\ee 
 Therefore, we have 
 \be
 \mathbb{P}\left(\mathrm{dist}\left(\hat{X}_0', \partial [\tilde{K}_{i} \cap \rho B]\right) < n^{-3}\right) \leq n^{-2} \sqrt{2} + \frac{\epsilon}{10n^2 p^2}+ \|\mathcal{L}(Z_k) - \pi_{\tilde{K}_{i}}\|_{\mathrm{TV}} \leq n^{-2} \sqrt{2} + \frac{\epsilon}{10n^2 p^2} + \frac{\epsilon}{p^2}.
 \ee 
  But $\hat{X}_0' \equiv \hat{X}_0'^k$ is generated independently at each iteration $k$ of the while loop, implying that 
  \be
  \mathbb{P}\left(\min_{j \leq k} \mathrm{dist}\left(\hat{X}_0'^j, \partial [\tilde{K}_{i} \cap \rho B]\right) \leq n^{-3}\right) \leq \left(n^{-2} \sqrt{2} + \frac{\epsilon}{10n^2 p^2} + \frac{\epsilon}{p^2}\right)^k \leq 2^{-k}.
  \ee 
   Therefore, the expected number of iterations of the While loop is bounded by $\sum_{k=1}^\infty 2^{-(k-1)} \leq 2$.
\end{proof}

\subsection{Bounding the expected running time of Algorithm \ref{alg:rounding}}\label{sec:ExpectedRunningTimeRounding}
In this section we bound the expected frequency at which Algorithm \ref{alg:rounding} checks any given inequalities.
From now on we set the parameter $\alpha$ of Algorithm \ref{alg:BallModified} to be $\alpha = 4\log\left(\frac{2 np  i_{\mathrm{max}}}{\epsilon} \right)$.

\begin{lemma} \label{lemma:ExpectedRuntime}
Suppose that, at each step of Algorithm \ref{alg:rounding}, $X_0$ is a $\beta$-warm start with respect to both the uniform distribution and the speedy distribution. 
 Fix  $b^{\star}\leq \frac{i_{\mathrm{max}}}{40 \mathcal{I}}$. 
  Then with probability at least $\frac{8}{10}$ the total number of inequality checks made during the first $b^{\star}$ times that Algorithm \ref{alg:rounding} with $\mathrm{Modified}= \mathrm{ON}$ invokes the ball walk Markov chain is $\leq 40 \beta b^{\star}\mathcal{I} \times m\left[16 n^{-1} \gamma +  \frac{32 \gamma}{n} \times \log(\nicefrac{n}{\gamma}) + \frac{\epsilon}{np}\right]$.
\end{lemma}

\begin{proof}
Let $\mathcal{S}^k$ be the number of proper+improper steps for the $k$th Markov chain in Algorithm \ref{alg:rounding}, and let $\mathcal{F}_j^k$ be the frequency at which the $k$th Markov chain checks the inequality $(A_j,b_j)$.

Let $\mathcal{S}^\dagger := \sum_{k=1}^{b^{\star}} \mathcal{S}^k$ be the sum of the steps for the first $b^{\star}$ runs of the Markov chain in Algorithm \ref{alg:rounding}, and let $ \mathcal{F}^\dagger := \frac{1}{b^{\star}}\sum_{k=1}^{b^{\star}} \sum_{j=1}^m \mathcal{F}_j^k$ be the frequency at which Algorithm \ref{alg:rounding} checks any inequality.

Now, if $Y_0$ is a $\beta$-warm start for the speedy distribution, by Lemmas \ref{lemma:NonproperSteps} and \ref{lemma:AverageConductance} we have
\be
\mathbb{E}[\mathcal{\mathcal{S}}^{\dagger}] \leq \beta \times 4 b^{\star}\mathcal{I}.
\ee
Hence, we have, by Markov's inequality, that $\mathcal{\mathcal{S}}^{\dagger} \leq 40 b^{\star}\mathcal{I} \leq i_{\mathrm{max}}$ with probability at least $\frac{9}{10}$.

Thus, by Lemma \ref{lemma:Frequency} we have 
\be
\mathbb{E}[\mathcal{F}^{\dagger} \mathbbm{1}\{\mathcal{S}^{\dagger} \leq i_{\mathrm{max}}\}]   \leq m \times \left[16n^{-1} \gamma \beta +  \frac{32 \gamma}{n} \times \beta \log\left(\nicefrac{n}{\gamma}\right) + \frac{\epsilon \beta}{n b^{\star}} \right].
\ee
Hence, by Markov's inequality we have with probability at least $\frac{9}{10}$ that 
\be
\mathcal{F}^{\dagger} \mathbbm{1}\{\mathcal{S}^{\dagger} \leq i_{\mathrm{max}}\} \leq  160 n^{-1} \gamma \beta +  \frac{320 \gamma}{n} \times \beta \log\left(\nicefrac{n}{\gamma}\right) + 10\frac{\beta \epsilon}{n b^{\star}}.
\ee
Therefore, with probability at least $\frac{8}{10}$ we have that both $\mathcal{\mathcal{S}}^{\dagger} \leq 40 \beta b^{\star}\mathcal{I} \leq i_{\mathrm{max}}$ and $\mathcal{F}^{\dagger} \leq 160n^{-1} \gamma \beta +  \frac{320 \gamma}{n} \times \beta \log\left(\nicefrac{n}{\gamma}\right) + 10\frac{\beta \epsilon}{n b^{\star}}$ (since $\mathcal{F}^{\dagger} \mathbbm{1}\{\mathcal{S}^{\dagger} \leq i_{\mathrm{max}}\} = \mathcal{F}^{\dagger}$ whenever $\mathcal{\mathcal{S}}^{\dagger} \leq 40 \beta b^{\star}\mathcal{I} \leq i_{\mathrm{max}}$).

Therefore, with probability at least $\frac{8}{10}$ we have that the total number of times Algorithm \ref{alg:BallModified} computes an inequality is
\be
\mathcal{\mathcal{S}}^{\dagger} \mathcal{F}^{\dagger} \leq  40 \beta b^{\star}\mathcal{I} \times m\left[16n^{-1} \gamma \beta +  \frac{32 \gamma}{n} \times \beta \log\left(\nicefrac{n}{\gamma}\right) + \frac{\beta \epsilon}{n b^{\star}}\right].
\ee
\end{proof}

\subsection{Bounding the warmness of the start in Algorithm \ref{alg:rounding} with respect to the speedy and uniform distributions}\label{sec:warmness}
\begin{lemma}\label{lemma:warmness}
With probability at least $\frac{6}{10} -\epsilon$ we have that, the total number of inequality checks for the duration of Algorithm \ref{alg:rounding} with $\mathrm{Modified}= \mathrm{ON}$, is at most $220 (2+p) i^\star \mathcal{I} \times m\left[16 n^{-1} \gamma +  \frac{32 \gamma}{n} \times \log\left(\nicefrac{n}{\gamma}\right) + \frac{\epsilon}{np}\right]$. 
 Moreover, if $\hat{\Sigma}_{i^\star}, \hat{\mu}_{i^\star}, X_0$ are the outputs of Algorithm \ref{alg:rounding}, we have that $X_0$ is O(1)-warm with respect to the uniform distribution on $\hat{\Sigma}_{i^\star}^{-\frac{1}{2}}(K - \hat{\mu}_{i^\star})$.
\end{lemma}

\begin{proof}
First, we bound the number of times Algorithm \ref{alg:BallModified} runs the ball walk. 
 By Lemma \ref{lemma:isotropic} we have that the expected number of iterations of the ``While" loop at each ``For" loop iteration of Algorithm \ref{alg:BallModified} is at most $2$. 
  Hence, the expected number of times the ball walk is run by Algorithm \ref{alg:BallModified} is at most $(2+p) i^\star$. 
   Therefore, by Markov's inequality, we have that with probability at least $\frac{9}{10}$ the ball walk is run by Algorithm \ref{alg:BallModified} no more than $10(2+p) i^\star$ times.

Next, we bound the speedy and uniform warmness. 
 The main observation of this section is that, for our step size, any distribution which is $\beta$-warm with respect to the uniform distribution is $2\beta$-warm with respect to the speedy distribution. 
  This is true since, by  Remark 4.12 in \cite{kannan1997random}, one can obtain speedy-distributed samples by starting with uniformly distributed samples, if we take one step of the ball walk starting at each sample, and reject the original point if and only if the ball walk step leaves the convex body. 
   By Lemma \ref{lemma:AverageConductance}, for our choice of step size $\eta$ the average acceptance probability $\lambda$ is at most $\frac{1}{2}$. 
    Therefore, any sample which is $\beta$-warm with respect to the uniform distribution must also be $2\beta$-warm with respect to the speedy distribution (since then $\hat{\pi}_K(x) = \pi_K(x) \frac{1}{\lambda}\mathbb{P}(\{\textrm{one step of ball walk starting at x is rejected}\})$) \footnote{Note, however, that the converse is not true: $\beta$-warm with respect to the speedy distribution does not imply $2\beta$ warm with respect to the uniform distribution.}. 
     Therefore, to apply Lemma \ref{lemma:ExpectedRuntime}, it is enough to show that the initial points in Algorithm \ref{alg:BallModified} are $\beta$-warm with respect to the uniform distribution.

Denote the value of $\hat{X}_0'$ at the $i$th iteration of the ``For" loop and k'th iteration of the ``While" loop of Algorithm  \ref{alg:rounding} by $\hat{X}_0' \equiv \hat{X}_0'^{i,k}$. 
  Next, we note that by the proof of Lemma \ref{lemma:isotropic} we have $\|\mathcal{L}(\hat{X}_0'^{i,k}) - \pi_{\tilde{K}_{i}}\|_{\mathrm{TV}}  \leq \frac{\epsilon}{10 p (i^{\star})^2}$ for each $i,k$.  Moreover (since we are using a fixed starting point) $\hat{X}_0'^1, \hat{X}_0'^2,\ldots$ are independent. 
   Thus, there exists a sequence of random variables $\{\tilde{X}_0^{i,k}\}_{i\in [i^\star], k \in [10(2+p) i^\star]}$ such that $\tilde{X}_0^{i,k} = \hat{X}_0'^{i,k}$ with probability at least $\frac{9}{10} -\epsilon$, where $\tilde{X}_0^{i,k} \sim \pi_{\tilde{K}_{i}}$. 
    Hence, $\tilde{X}_0^{i,k}$ are 1-warm with respect to the uniform distribution on $\tilde{K}_{i}$.

Recall from the proof of Lemma \ref{lemma:isotropic} that $K_{i} \subseteq K_{i+1}$ and  $\mathrm{Vol}(K_{i+1}) \leq e\mathrm{Vol}(K_{i})$. 
 Thus,
 \be
 \tilde{K}_{i} \subseteq \hat{K}_{i+1} \quad \textrm{  and } \quad \mathrm{Vol}(\hat{K}_{i+1}) \leq e \mathrm{Vol}(\tilde{K}_{i}),
 \ee
 since $\tilde{K}_{i} := \hat{\Sigma}_{i+1}^{-\frac{1}{2}}(K_{i}-  \hat{\mu}_{i+1})$ and $\hat{K}_{i+1} := \hat{\Sigma}_{i+1}^{-\frac{1}{2}} (K_{i+1} - \hat{\mu}_{i+1})$).

Therefore, the fact that $\tilde{X}_0^{i,k}$ is 1-warm with respect to the uniform distribution on $\tilde{K}_{i}$ implies that it is $e$-warm with respect to the uniform distribution on  $\hat{K}_{i+1}$.  Hence, we have that $\tilde{X}_0^{i,k}$ is $2e$-warm with respect to the speedy distribution on $\hat{K}_{i+1}$.

Hence, by Lemma \ref{lemma:ExpectedRuntime} we have that with probability at least $\frac{6}{10} -\epsilon$, the total number of inequality checks for the duration of Algorithm \ref{alg:rounding}, is at most $80 e(2+p) i^\star \mathcal{I} \times m\left[16 n^{-1} \gamma +  \frac{32 \gamma}{n} \times \log\left(\nicefrac{n}{\gamma}\right) + \frac{\epsilon}{np}\right]$.
\end{proof}

\subsection{Counting the running time of subroutines with negligible contributions to the running time}\label{sec:NegligibleContributions}
\paragraph{Sorting.}
First, we argue that the cost of sorting the $h_j$'s can be ignored.

Algorithm \ref{alg:BallModified} must sort the $m$ inequalities into $O(\frac{n}{\alpha \eta}) = O(n^{1.5})$ bins, since, if $K$ is $O(1)$-isotropic, the diameter of $K$ is $O(\sqrt{n})$. 
 This takes $O(m \log(n))$ time. 
 This only occurs once, at the start of the algorithm.  Since our bound for the entire algorithm is $O(m n^{4.5})$ and $m n^{4.5}>>m \log(n)$, we can ignore this cost.

At each step of the Markov chain, Algorithm \ref{alg:BallModified} must move all the bins over by 1 to ``add" $\frac{n}{\alpha \eta}$ to each bin, and select all $h_j$'s such that $h_j < \alpha \frac{\eta}{\sqrt{n}} \times i$. 
 This can be done in $O(1)$ operations by simply moving a ``pointer" at the bin for elements in $[\alpha \frac{\eta}{\sqrt{n}} \times (i-1), \alpha \frac{\eta}{\sqrt{n}} \times i]$ one bin to the right. 
  This is only done once every Markov chain step, and is negligible in comparison to the cost of computing a Markov chain step, implying that we can ignore this $O(1)$ cost.

After computing the new values for all the $h_j$'s that were selected, Algorithm \ref{alg:BallModified} must then sort these $h_j$s into the corresponding bins. 
 Since there are $O(n^{1.5})$ bins, this takes $O(\log(n))$ time for each selected $h_j$. 
  This is negligible in comparison to the cost of recomputing the value of each $h_j$, which is $O(n)$ (since we have to take an inner product). 
   Therefore we can ignore this cost.

\paragraph{Applying the linear transformation to $K_i$.}
Next, we argue that the cost of applying the linear transformation to put $K_i$ into isotropic position at each iteration of Algorithm \ref{alg:rounding} can be ignored. 
 This linear transformation requires computing $A \hat{\Sigma}_i^{\frac{1}{2}}$, which takes $O(mn^2)$ operations,  and computing $b -A \hat{\mu}_i$, which takes $O(mn)$ operations. 
  There are $i^\star= n\log(\frac{R}{r})$ iterations in Algorithm \ref{alg:rounding}, so applying the linear transformations contributes at most $O(mn^3 \log(\frac{R}{r}))$ arithmetic operations. 
   This is much smaller than our bound of $O(mn^{4.5} \log(\frac{R}{r}))$ on the number of operations, so we can ignore the cost of applying the linear transformations.

\subsection{Proof of main theorem for rounding} \label{sec:proofrounding}
We now consider the following procedure, where we run Algorithm \ref{alg:rounding} multiple times until it succeeds. 
 This allows us to put the polytope into isotropic position with a bound on the running time that holds with very high probability (as opposed to just holding in expectation).
\begin{algorithm}[H]
\caption{Rounding in bounded time} \label{alg:roundingBoundedTime}
\flushleft
\textbf{input:} $A \in \mathbb{R}^{m \times n}$, $b \in \mathbb{R}^m$\\
\textbf{input:} $r,R>0$ such that $r B\subseteq K \subseteq R B$, where $K := \{x \in \mathbb{R}^n : A x \leq b\}$\\
\textbf{input:} $p\in \mathbb{N}$, $\epsilon>0$\\
\textbf{input:} $\mathfrak{I}, i_{\mathrm{max}}\in \mathbb{N}$
\begin{algorithmic}[1]
\State Set $\mathrm{Success} = \mathrm{False}$
\While{$\mathrm{Success} = \mathrm{False}$}
\State  Run Algorithm \ref{alg:rounding} with $\mathrm{Modified}= \mathrm{ON}$ and with the above inputs until either it outputs $\hat{\Sigma}_{i^\star}, \hat{\mu}_{i^\star}$, $X_0$ or until it completes $\mathfrak{I}$ inequality checks.
\State If $\hat{\Sigma}_{i^\star}, \hat{\mu}_{i^\star}, X_0$ were obtained, output $\mathrm{Success} = \mathrm{True}$
\EndWhile\\
\textbf{output:} $\hat{\Sigma}_{i^\star}, \hat{\mu}_{i^\star}$, $X_0$
\end{algorithmic}
\end{algorithm}

\noindent
From now on we choose  $\mathfrak{I}, i_{\mathrm{max}} $ to satisfy
\be\label{eq:repeat}
\mathfrak{I} &=  2000 (2+p) i^\star \mathcal{I} \times m\left[16 n^{-1} \gamma +  \frac{32 \gamma}{n} \log\left(\nicefrac{n}{\gamma}\right) + \frac{6 \epsilon}{np}\right],\\
i_{\mathrm{max}} &= \left(2000 (2+p) i^\star \mathcal{I} \times m\left[16 n^{-1} \frac{20m\hat{\eta}}{\epsilon} +  \frac{640 m\hat{\eta}}{n\epsilon} \log\left(\frac{n \epsilon}{20m\hat{\eta}}\right) + \frac{\epsilon}{np}\right] \right)^2.
\ee
In particular, we have $i_{\mathrm{max}} \geq \mathfrak{I}$.

\begin{lemma}\label{lemmaRepeat}
Fix $\epsilon>0$. 
 Then with probability at least $1-2\epsilon\log(\frac{1}{\epsilon})$ Algorithm   \ref{alg:roundingBoundedTime}  outputs, after at most $\log(\frac{1}{\epsilon})$ calls to Algorithm \ref{alg:rounding}, an affine transformation $(\hat{\Sigma}_{i^\star}, \hat{\mu}_{i^\star})$ and a point $X_0$, for which $K^\dagger := \hat{\Sigma}_{i^\star}^{-\frac{1}{2}}(K- \hat{\mu}_{i^\star})$ is in 2-isotropic position. 
  Moreover we have that $X_0$ is $O(\log(\frac{1}{\epsilon}))$-warm with respect to the uniform distribution on $K^\dagger$.
\end{lemma}
\begin{proof}
First, we note that by Lemma \ref{lemma:warmness} we have that the ``While" loop of Algorithm \ref{alg:roundingBoundedTime} completes at least one run after $\log(\frac{1}{\epsilon})$ runs with probability at least $1- 2^{-\log(\frac{1}{\epsilon})} \geq 1- \epsilon$. 
 Moreover, by Lemmas \ref{lemma:isotropic} and \ref{lemma:warmness}, if Algorithm \ref{alg:rounding} were allowed to keep running even after it uses up its alotted number of arithmetic operations and ball walk steps, it would return an affine transformation which puts $K$ into 2-isotropic position, and an $O(1)$-warm start for the uniform distribution on this affine transformation of $K$, with probability at least $1-\epsilon$ at each run. 
  Therefore, if we do stop Algorithm \ref{alg:rounding} after its alotted number of arithmetic operations and ball walk steps, after $\log(\frac{1}{\epsilon})$ steps we obtain, with probability at least $1-\epsilon-\epsilon \log(\frac{1}{\epsilon})$, an affine transformation $(\hat{\Sigma}_{i^\star}, \hat{\mu}_{i^\star})$ for which $K^\dagger := \hat{\Sigma}_{i^\star}^{-\frac{1}{2}}(K- \hat{\mu}_{i^\star})$ is in 2-isotropic position and an $O(\log(\frac{1}{\epsilon}))$-warm start $X_0$ for the uniform distribution on $K^\dagger$.
\end{proof}
Fix $p= n \times c \log^2(\frac{1}{\epsilon}) \log^2(n)$. Recall the shorthand notation $\hat{\eta} := \frac{1}{10}\eta \sqrt{n}$, and $\gamma:=10\alpha \hat{\eta}$,  and that we have fixed the parameters $p$, $\mathcal{I}$, $\eta$, $\alpha$ of Algorithm \ref{alg:BallModified} as follows:
\begin{itemize}
\item $\mathcal{I} = c_2 n^{2} 20 \sqrt{n}\log\left(\frac{40n^2 p^2}{\epsilon}\right) \left(\log \log 20 \sqrt{n}\log\left(\frac{40n^2 p^2}{\epsilon}\right)\right) \log^3\left(\frac{np^2}{\epsilon}\right) \log \log\left(\frac{R}{r}\right)$,
\item $\eta = \frac{1}{30\sqrt{n\log(\nicefrac{n}{\epsilon})}}$,
\item $\alpha = 4\log\left(\frac{2 np  i_{\mathrm{max}}}{\epsilon} \right)$.
\end{itemize}
Also recall that we have fixed $\mathfrak{I}$ and $i_{\mathrm{max}}$ in Equation \ref{eq:repeat}, and that $i^\star$ is set to $i^\star= n\log_2(\frac{R}{r})$ in Algorithm \ref{alg:rounding}.
\begin{theorem}[\bf Version of Main Theorem (Th. \ref{thm:main_intro}) specific to our algorithm] \label{thm:main_AlgorithmSpecific}
Fix $\epsilon>0$. 
 With probability at least $1-2\epsilon\log(\frac{1}{\epsilon})$ Algorithm \ref{alg:roundingBoundedTime} outputs an affine transformation $(\hat{\Sigma}_{i^\star}, \hat{\mu}_{i^\star})$ and a point $X_0$ for which $K^\dagger := \hat{\Sigma}_{i^\star}^{-\frac{1}{2}}(K- \hat{\mu}_{i^\star})$ is in 2-isotropic position, and $X_0$ is $O(\log(\frac{1}{\epsilon}))$-warm with respect to the uniform distribution on $K^\dagger$, in  $\tilde{O}(mn^{4.5} \log^9(\frac{1}{\epsilon}) \log^9(\frac{R}{r}))$ arithmetic operations. 
\end{theorem}
\begin{proof}
By Lemma \ref{lemmaRepeat}, with probability at least $1-2\epsilon\log(\frac{1}{\epsilon})$ Algorithm \ref{alg:roundingBoundedTime} outputs an affine transformation $(\hat{\Sigma}_{i^\star}, \hat{\mu}_{i^\star})$ which puts $K$ into 2-isotropic position, in at most $\mathfrak{I}\log(\frac{1}{\epsilon})$ inequality checks. 
 The number of arithmetic operations for each inequality check is no more that $3n$, implying that the number of arithmetic operations performed by Algorithm \ref{alg:roundingBoundedTime} in this event is at most 
 \be
 3n \times \mathfrak{I}\log(\frac{1}{\epsilon}) = 6000 (2+p) i^\star \mathcal{I} \times m\left[16 \gamma +  32 \gamma\log(\nicefrac{n}{\gamma}) + \frac{6 \epsilon}{p}\right]\log(\frac{1}{\epsilon}) = \tilde{O}(mn^{4.5} \log^9(\frac{1}{\epsilon}) \log^9(\frac{R}{r})).
 \ee
\end{proof}

\section{Proofs of corollaries for volume computation and sampling}\label{sec:CorollaryProofs}
In this section we prove Corollaries \ref{thm:intro_volume} and \ref{cor:sampling} from the introduction.  First, we prove Corollary \ref{thm:intro_volume} on volume estimation.
\begin{proof}[Proof of Corollary \ref{thm:intro_volume}]
By Lemma \ref{thm:main_AlgorithmSpecific} we have that Algorithm \ref{alg:roundingBoundedTime} obtains an affine transformation $(\hat{\Sigma}_{i^\star}, \hat{\mu}_{i^\star})$ such that $K^\dagger := \hat{\Sigma}_{i^\star}^{-\frac{1}{2}}(K - \hat{\mu}_{i^\star})$ is in 2-isotropic position, in $\tilde{O}(mn^{4.5} \mathrm{polylog}(\frac{1}{\epsilon}, \frac{R}{r}))$ arithmetic operations. 
 Thus, $K^\dagger$ contains a ball $B(\mu_{K^\dagger},\frac{1}{2})$ where $\|\mu_{K^\dagger}\|_2 \leq \frac{1}{5}$, implying that $B(0,\frac{1}{4}) \subseteq K^\dagger$. 
  Thus $4K^\dagger$ contains the unit ball $B$ and is 8-isotropic. 
   Hence, if we then apply the volume algorithm \cite{cousins2015bypassing} to $4K^\dagger$, by Theorem 1.1 of \cite{cousins2015bypassing} we can compute the volume of $4K^\dagger$ up to a factor of $1+\delta$ in $\frac{mn^4}{\delta^2}\log^6(\frac{n}{\delta})\log(\frac{1}{\epsilon})$  arithmetic operations. 
    But 
    $$\mathrm{Vol}(K) = \frac{1}{4^n}\det(\hat{\Sigma}_{i^\star}^{\frac{1}{2}})\mathrm{Vol}(4K^\dagger).$$
     Since $\det(\hat{\Sigma}_{i^\star}^{\frac{1}{2}})$ can be computed in $O(n^3)$ arithmetic operations, we can compute with probability at least $1-\epsilon$ the volume of $K$ up to a factor of $1+\delta$ in $\tilde{O}(mn^{4.5} \mathrm{polylog}(\frac{1}{\epsilon}, \frac{R}{r}) + \frac{mn^4}{\delta^2}\log^6(\frac{n}{\delta})\log(\frac{1}{\epsilon}))$ arithmetic operations.
\end{proof}

\noindent We then prove Corollary \ref{cor:sampling} on sampling:
\begin{proof}[Proof of Corollary \ref{cor:sampling}]
By Theorem \ref{thm:main_AlgorithmSpecific} we have that Algorithm \ref{alg:roundingBoundedTime}  obtains an affine transformation $(\hat{\Sigma}_{i^\star}, \hat{\mu}_{i^\star})$ and a random vector $X_0$ such that $K^\dagger := \hat{\Sigma}_{i^\star}^{-\frac{1}{2}}(K - \hat{\mu}_{i^\star})$ is in 2-isotropic position and $X_0$ is $O(1)$-warm with respect to the uniform distribution on $K^\dagger$ in $\tilde{O}(mn^{4.5} \mathrm{polylog}(\frac{1}{\epsilon}, \frac{R}{r}))$ arithmetic operations.

Since $K^\dagger$ is in 2-isotropic position it is contained in $B(\mu_{K^\dagger},4n)$ where $\|\mu_{K^\dagger}\|_2 \leq \frac{1}{5}$. 
 Moreover, by Lemma 24 in \cite{lee2017eldan}, $1-\epsilon$ of the volume of the convex body $K^\dagger$ is contained in a ball of radius $2\sqrt{n} \log(\frac{1}{\epsilon})$ centered at $\mu_{K^\dagger}$.

Hence, by Theorem 1.1 of \cite{lovasz2006hit}, if we the run the hit-and-run Markov chain sampling algorithm for  $\tilde{O}(n^3 \log(\frac{1}{\epsilon}))$ Markov chain steps on the convex body $K^\dagger \cap B(0, 4\sqrt{n} \log(\frac{1}{\epsilon}))$ with O(1)-warm start $X_0$, we obtain a point $Z$ uniformly distributed on $K^\dagger  \cap B(0, 4\sqrt{n} \log(\frac{1}{\epsilon}))$ with TV error $\epsilon$.

But since $\mathrm{Vol}(K^\dagger \cap B(0, 4\sqrt{n} \log(\frac{1}{\epsilon}))) \geq (1-\epsilon)\mathrm{Vol}(K^\dagger)$, we have that  $Z$ is also uniformly distributed on $K^\dagger$ with TV error $2\epsilon$.  To obtain a point on $K$, we compute  $\tilde{Z}= \hat{\Sigma}_{i^\star}^{\frac{1}{2}}Z + \hat{\mu}_{i^\star}$. 
 Since  $Z$ is uniformly distributed on $K^\dagger$ with TV error $2\epsilon$, $\tilde{Z}$ must be uniformly distributed on $K$ with TV error $2\epsilon$ as well.

Therefore, we obtain a point $Z$ from the uniform distribution on $K$ with TV error $2\epsilon$, in a number $\tilde{O}(mn^{4.5} \mathrm{polylog}(\frac{1}{\epsilon}, \frac{R}{r}) + n^3 \log(\frac{1}{\epsilon})) = \tilde{O}(mn^{4.5}\mathrm{polylog}(\frac{1}{\epsilon}, \frac{R}{r}))$ of arithmetic operations.
\end{proof}

\newpage

\bibliographystyle{plain}
\bibliography{ManyConstraints}

\end{document}